\documentclass[11pt]{article}
\usepackage{float}
\usepackage{listings}
\usepackage{algorithm}
\usepackage[noend]{algpseudocode}
\usepackage{amsmath, amsthm, amssymb, amsfonts}
\usepackage{graphicx}
\usepackage{color}
\usepackage{xy}
\usepackage{url}
\usepackage{caption}
\usepackage{multirow}
\usepackage{wrapfig}
\usepackage{fullpage}

\newtheorem{theorem}{Theorem}
\newtheorem{corollary}{Corollary}

\newtheorem{proposition}{Proposition}
\newtheorem{definition}{Definition}

\newcommand{\set}[1]{\{#1\}} 

\newcommand{\CommentLine}[1]{
    \State $\triangleright$ {#1}
}

\DeclareMathOperator*{\E}{\mathbb{E}}

\title{Using Skip Graphs for Increased NUMA Locality}

\author{Samuel Thomas, Ana Hayne, Jonad Pulaj, Hammurabi Mendes\\ Davidson College, NC, USA\\ \texttt{sathomas,anhayne,jopulaj,hamendes}@davidson.edu}

\date{}

\begin{document}

\maketitle

\begin{abstract}
We present a data partitioning technique performed over skip graphs that promotes significant quantitative and qualitative improvements on NUMA locality in concurrent data structures, as well as reduced contention.
We build on previous techniques of thread-local indexing and laziness, and, at a high level, our design consists of a partitioned skip graph, well-integrated with thread-local sequential maps, operating without contention. As a proof-of-concept, we implemented map and relaxed priority queue ADTs using our technique. Maps were conceived using lazy and non-lazy approaches to insertions and removals, and our implementations are shown to be competitive with state-of-the-art maps. We observe a 6x higher CAS locality, a 68.6\% reduction on the number of remote CAS operations, and a increase from 88.3\% to 99\% CAS success rate when using a lazy skip graph as compared to a control skip list (subject to the same codebase, optimizations, and implementation practices). Qualitatively speaking, remote memory accesses are not only reduced in number, but the larger the NUMA distance between threads, the larger the reduction is.

We consider two alternative implementations of relaxed priority queues that further take advantage of our data partitioning over skip graphs: (a) using ``spraying'', a well-known random-walk technique usually performed over skip lists, but now performed over skip graphs; and (b) a custom protocol that traverses the skip graph deterministically, marking elements along this traversal. We provide formal arguments indicating that the first approach is more \emph{relaxed}, that is, that the span of removed keys is larger, while the second approach has smaller contention. Experimental results indicate that the approach based on spraying performs better on skip graphs, yet both seem to scale appropriately.
%
\end{abstract}

\section{Introduction}

The increasing availability of computing cores on shared memory machines makes concurrent data structure design a critical factor for application performance. Non-blocking\cite{progress}, linearizable \cite{linearizability} structures are particularly appealing, since they can effectively replace sequential or blocking (lock-based) structures without compromising the semantics expected by the programmers. However, the design landscape for concurrent structures is changing: NUMA architectures emerge as a set of computing/memory ``nodes'' linked by an interconnect, making memory accesses within the same NUMA node cheaper than those made across different NUMA nodes.
 
Under the usual assumption that threads are pinned to cores, we adopt the definition of \emph{local memory} accesses as those operating in memory initially allocated by the current thread (under first-touch NUMA policy), and \emph{remote} accesses as those accesses that are non-local. Note that our definitions are conservative, as data initially allocated by two different threads could indeed be located within the same NUMA node. Our goal is to increase NUMA \emph{locality} -- the ratio of local over total memory accesses, and research is very active in this area. Some approaches \cite{cphash, numastack, numask} focus on redesigning data structures with NUMA awareness, which is effective as we have full ability to exploit the structure's internal features for the task. Unfortunately, complete redesigns can pose significant development and research efforts, unsuitable for non-specialists. On the other hand, approaches like~\cite{blackbox} allow sequential structures to be ``plugged-in'' and benefit from NUMA-aware concurrency, based on replicating the dataset among nodes, batching local operations, and coordinating batches as to minimize inter-node traffic.

\textbf{Our contributions.} We present a data partitioning technique based on using \emph{skip graphs}~\cite{skipgraphs,skipnets} in order to promote significant quantitative and qualitative improvements in NUMA locality in concurrent data structures, as well as reduced contention. By qualitative improvement, we mean that remote memory accesses are not only reduced in number, but the larger the NUMA distance between threads, the larger the reduction is (Sec.~\ref{Sec-Evaluation}). At a high level, skip graphs can be viewed as multiple skip lists that overlap in a structured way in multiple levels, so we partition sub-components of those skip graphs among threads, taking advantage of its structure, in a way that makes the dataset to be operated with much higher NUMA locality. Skip graphs are expensive data structures, so our technique is made viable in practice by using thread-local indexing and well-documented laziness principles \cite{numask,nohotspot} into our design.
As a proof-of-concept of our data partitioning technique, we implemented maps and relaxed priority queues \cite{pqrelaxed3, pqrelaxed4,pqrelaxed2,pqrelaxed1}. Maps have been implemented with and without using laziness techniques, but always using thread-local indexing similarly to \cite{numask}. We are competitive with state-of-the-art maps \cite{numask,nohotspot,rotating}: in some cases, we see 80\% increased performance, while in others, we see similar performance to the faster running implementation. As part of our NUMA locality assessment, we observe a 6x higher CAS locality, a 68.6\% reduction on the number of remote CAS operations, and an increase from 88.3\% to 99\% of CAS success rate when using a lazy skip graph map implementation, as compared to our \emph{control} -- a skip list subject to the same codebase, optimizations, and implementation practices. Memory access patterns are visualized on Sec.~\ref{Sec-Evaluation}, and the qualitative improvement introduced with our technique is evident. 
 
We also contribute with \emph{relaxed priority queues} \cite{pqrelaxed3,pqrelaxed4,pqrelaxed2}, which return an element among the $k$ smallest elements in a set, rather than the absolute smallest. We not only use our data partitioning technique, which brings increased NUMA locality and reduced contention, but we consider a couple of algorithmic choices that further harness key structural features of the skip graph: (a) using \emph{spraying} \cite{pqrelaxed1}, a well-known random-walk technique usually performed over skip lists, but now performed over skip graphs; and (b) a custom protocol that traverses the skip graph deterministically, marking elements along this traversal. In either case, our NUMA locality improvements, discussed above, still apply, and we do verify experimentally that both approaches perform about 5x-9x faster than our control, which is a spray performed over a skip list, but one implemented with the same codebase and practices as the approaches (a) and (b) above. Additionally, we provide a \emph{formal argument} indicating that our custom protocol (b) is slightly less relaxed (Theorem.~\ref{thr:range_spray_sg}), that is, that the span of removed keys is slightly smaller (but still the same asymptotically), but our approach is subject to smaller contention: Theorems~\ref{thr:cas_slspray} and~\ref{thr:cas_sgspray} indicate that each node is subject to contention from only two threads in our custom algorithm.

Skip graphs as defined in \cite{skipgraphs,skipnets} have expensive insertions and removals compared to skip lists, related to the expected number of levels nodes insert/remove themselves to/from. Our skip graphs, however, can have a reduced height as we pair them with \emph{thread-local indexes} (used also in \cite{numask}). Furthermore, we also employ \emph{laziness} \cite{numask,nohotspot} in our variant implementations (maps and relaxed priority queues), as we build internal skip graph links only when necessary (and, symmetrically, we destroy internal links only when deemed appropriate). Using these techniques not only grants the performance benefits previously documented in the literature \cite{numask,nohotspot}, but make skip graphs \emph{viable} in the first place, which promotes \emph{further} benefits on increased NUMA locality. Among the most crucial benefits is the qualitative improvement on the memory access pattern among threads, discussed and visualized on Sec.~\ref{Sec-Evaluation}. We anticipate that qualitative improvement to become increasingly significant in future NUMA systems, as we conjecture increasingly larger hierarchical cost structures for remote memory accesses in upcoming NUMA systems.

%
We proceed with an overview and background in Sec.~\ref{Sec-ArchitectureOverviewBackground}, and related work on Sec.~\ref{Sec-RelatedWork}. Design and implementation are further discussed in Sec.~\ref{Sec-ImplementationDetails}, with evaluation in Sec.~\ref{Sec-Evaluation} and conclusion in Sec.~\ref{Sec-Conclusion}.

\section{Architecture Overview and Background}
\label{Sec-ArchitectureOverviewBackground}

Our \emph{layered structure} consists of many \emph{local structures}, which are thread-local, sequential, navigable indexes (e.g., a C++ \texttt{map}), one assigned per thread, as well as a single \emph{shared structure}, which is a skip graph (or a slight variant thereof).
The local structures are used to ``jump'' to positions in the shared structure near to where insertion, removal, and contain operations will complete, which certainly contributes to reducing remote memory accesses. Once in the shared structure, our data partitioning scheme over the skip graph promotes a further reduction in remote memory accesses due to its internal structure.
We say \emph{nodes} store \emph{elements}, although we use these terms interchangeably. When necessary, we further distinguish between \emph{local nodes} or \emph{shared nodes}, as we refer to nodes belonging to a local structure or to the shared structure, respectively. Each local structure's job is to map elements inserted by their owning thread to the corresponding shared node in the skip graph.
Inserting an element $e$ adds a shared node \lstinline{s} to the skip graph, and creates a mapping $e \rightarrow \mathtt{s}$ in the local structure of the inserting thread. A subsequent removal of element $e$ will (i) \emph{logically delete} the shared node \lstinline{s} in the skip graph, (ii) cause a physical cleanup to occur in the shared structure and (iii) cause the thread that contains the mapping $e \rightarrow \mathtt{s}$ in its local structure to physically cleanup that association upon detection. Steps (ii) and (iii) can happen in any order. Discussion of insertion and removal definitions can be found in Sec.~\ref{Sec-ImplementationDetails}.

\textbf{Skip Graphs.} Fig.~\ref{fig:overall_design} depicts a skip graph, and performs the role of the shared structure. It is composed of multiple singly-linked lists at different levels (Fig.~\ref{fig:overall_design}). Starting from level zero, each level $i$ contains $2^i$ lists. All elements belong to the level-0 list, labeled as ``$\lambda$'', the empty string. The level-0 list is partitioned into two level-1 lists, labeled ``0'' and ``1''. Each level-1 list is further partitioned into two level-2 lists: the level-1 list labeled ``0'' (resp. ``1'') is partitioned in two level-2 lists, labeled ``00'' and ``01'' (resp. ``10'' and ``11''). The specific partitioning of elements (described ahead) is \emph{not} done in a probabilistic way as in the original skip graph: we have a \emph{partitioning scheme} that assigns threads to levels.
\begin{figure}[htb]
	\centering
	\includegraphics[width=.8\textwidth]{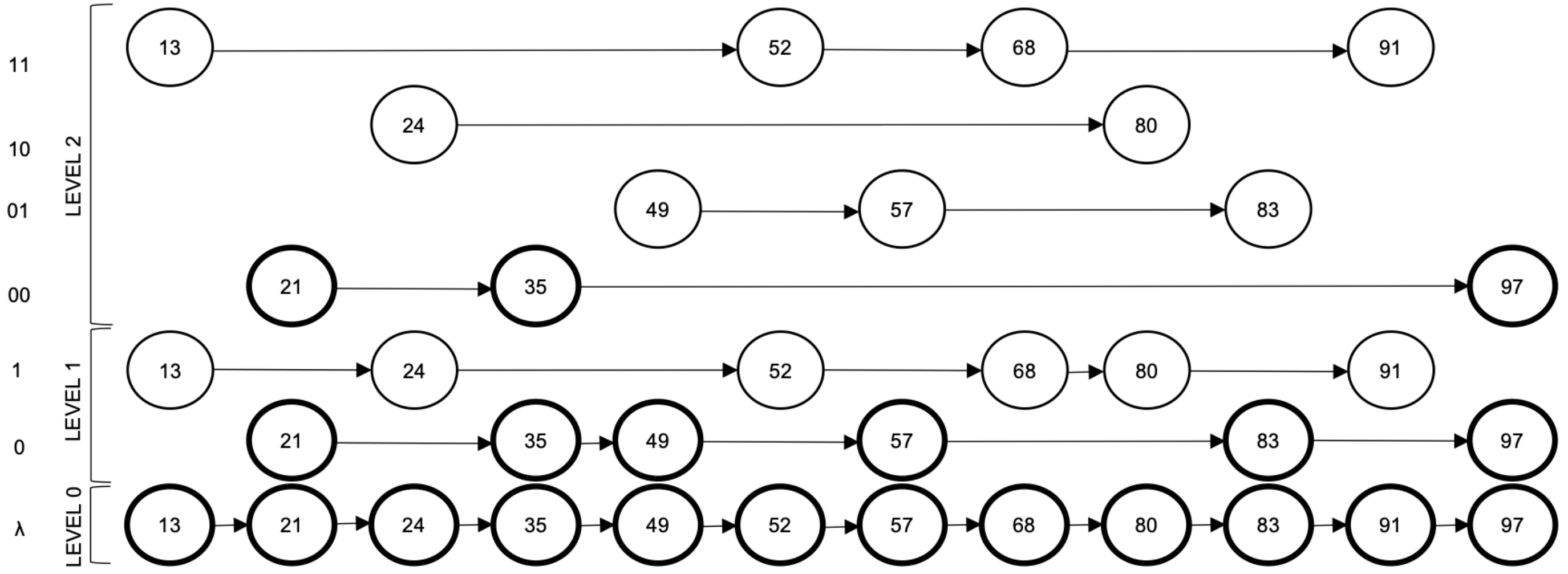}
	\caption{A skip graph can be seen as a set of skip lists sharing their levels (one of them is highlighted). It can also be seen as a set $2^i$ linked lists at each level $i$. With $T$ threads, $T/2^i$ of them work in each level-$i$ linked list, reducing contention and increasing locality.}
	\label{fig:overall_design}
\end{figure}

A skip graph can be seen as a collection of skip lists sharing their levels. In Fig.~\ref{fig:overall_design}, we denote by $H = (\lambda, 0, 00)$ a skip list formed by the highlighted nodes, one among four other skip lists that together constitute the skip graph. The element 35 is in the maximum level of $H = (\lambda, 0, 00)$. We refer to any skip list within the skip graph as \emph{shared skip list}, and any of the individual linked lists within the skip graph as a \emph{shared linked list}. All elements belong to one shared skip list in all of its levels, so we can \emph{always} perform a skip list search starting from that node's top level, traversing only the skip list that the node belongs to in its top level. For instance, in Fig.~\ref{fig:overall_design}, we can search for $83$ starting at $52$, following the shared skip list $(\lambda, 1, 11)$ through the following path: 52 $\rightarrow$ 68, $\downarrow$, 68 $\rightarrow$ 80, $\downarrow$, 80 $\rightarrow$ 83. Hence, despite a richer structure, it is important to notice that skip graph searches are skip list searches.

\textbf{Data Partitioning.} Even though the skip graph is shared by all threads, we limit where each thread operates on it, configuring a \emph{partition scheme}. We consider a system of $T$ threads $\set{t_0 \ldots t_{T-1}}$. The skip graph's maximum level is denoted by $\mathrm{MaxLevel}$, and we have that $\mathrm{MaxLevel} = \lceil \log(T) \rceil - 1$. Such a ``low'' $\mathrm{MaxLevel}$ does not guarantee logarithmic searches on a skip graph standing by itself, but our local structures compensate this fact by ``jumping'' to positions in the skip graph near where operations take place. Each thread $T_i$ will have a \emph{membership vector} $M_i$, a sequence of $\mathrm{MaxLevel}$ bits, whose suffixes indicate the \emph{only} shared skip list that $T_i$ can operate, which we denote $L_i$ and also call the \emph{associated skip list} of $T_i$. This causes all insertions from a single thread to happen in a unique skip list within the skip graph. As an example, consider the skip graph with $\mathrm{MaxLevel} = 2$ of Fig.~\ref{fig:overall_design}. Each thread will have a 2-bit membership vector. If $T_i$ has $M_i =$ ``10'', say, then $T_i$ will always insert or remove in the skip list $L_i = (\lambda, 0, 10)$. At most two threads operate in each of the top-level shared linked lists, and at most $T/2^i$ threads operate in any particular level-$i$ list.

Note that the original definition of skip graphs assumes a membership vector \emph{per element}, but our approach has all elements inserted by a single thread to follow the \emph{thread's membership vector}. Further, membership vectors are generated according to \emph{physical NUMA features} of the machine (see Sec.~\ref{Sec-Evaluation}). We make threads that run on nearby CPUs share more linked lists than threads that run on distant CPUs in the NUMA system. For example, consider a system where we have $T = 16$ threads ($\mathrm{MaxLevel} = 3$), with 2 NUMA nodes, each with 2 CPUs, each of those with 2 cores, each of those with 2 hyperthreads. We will give the same membership vector for any two threads running on hyperthreads of the same core; membership vectors with a common 2-bit suffix if threads run on different cores, and with a common 1-bit suffix if threads run on different CPUs; and finally membership vectors with no common suffix if threads run on different NUMA nodes. Now, on the level-3 linked lists, any contention relates to core-local data (and hopefully located on the core's closest cache); on the level-2 linked lists, any contention relates to CPU-local data; on the level-1 lists, any contention relates to NUMA-local data. Not only we expect less contention in upper-level lists, because they are shared among less threads, but we expect this \emph{contention to relate to more local data}. Further, any search that traverses the skip graph, as we discussed, is a skip list search. Hence, we \emph{first} traverse core-local data, and \emph{if we go down a level, the target data cannot be located in the same core}. Similarly, we then traverse CPU-local data, then NUMA-local data, and if we ever go down a level, the target data must be found remotely, where local/remote depend on the level in question. If we ever leave, say, a NUMA node, we never come back to it. The same applies to CPU-local data, or core-local data, and this happens \textbf{as a direct consequence of our data partitioning mechanism}, based on careful attribution of membership vectors. We have an automated mechanism that generates membership vectors based on inspecting the system’s CPU/socket/domain structure, and we indeed verify less contention (substantial improvement on CAS success ratio) and more locality (visualized graphically) in Sec.~\ref{Sec-Evaluation}.
\label{partitioning}

\textbf{Alternative shared structure: Sparse Skip Graphs.} In order to further explore benefits and tradeoffs of skip graphs, we also created and tested a second shared structure, called a \emph{sparse skip graph} (Fig.~\ref{fig:sparseskipgraph}). This structure is a skip graph where elements are present in level $i$ \emph{of any shared skip list} with expectation $1/{2^i}$, just like in a regular skip list.

The sparse skip graph is still a set of skip lists sharing their levels, although the levels overall become more and more sparse like a skip list. The combination of skip graph partitioning and skip list refinement makes elements be present in level $i$ \emph{of a particular linked list} with expectation $1/{4^i}$. 
For instance, in Fig.~\ref{fig:overall_design}, each of the level-1 lists ``0'' and ``1'' would partition only 50\% of the elements of ``$\lambda$'', which would be selected at 50\% chance independently. So, ``0'' and ``1'' would each have about 25\% of the elements in ``$\lambda$''. Similarly, the level-2 lists ``00'' and ``01'' would partition only 50\% of the elements of ``0'', each selected at 50\% chance, independently. So, lists ``00'' and ``01'' would each have about 6.25\% of the elements in ``$\lambda$''. Importantly, in our technique, only elements that reach the top level are added to the local structures. Therefore, sparse skip graphs also cause the local structures to become more sparse. This is crucial because the local structures, besides pointing to shared nodes nearby the target destinations, should also point to maximum-level nodes from which we can start an efficient search. Hence, using sparse skip graphs gives two immediate advantages: (i) the local structures are smaller; and (ii) the insertion and removal in the shared structure requires changes in less than $\mathrm{MaxLevel}$ levels. The tradeoff is that the starting point given by the local structures is not as close to the requested element compared to regular skip graphs.

\begin{figure}[htb]
	\centering
	\includegraphics[width=4.5cm,width=.8\textwidth]{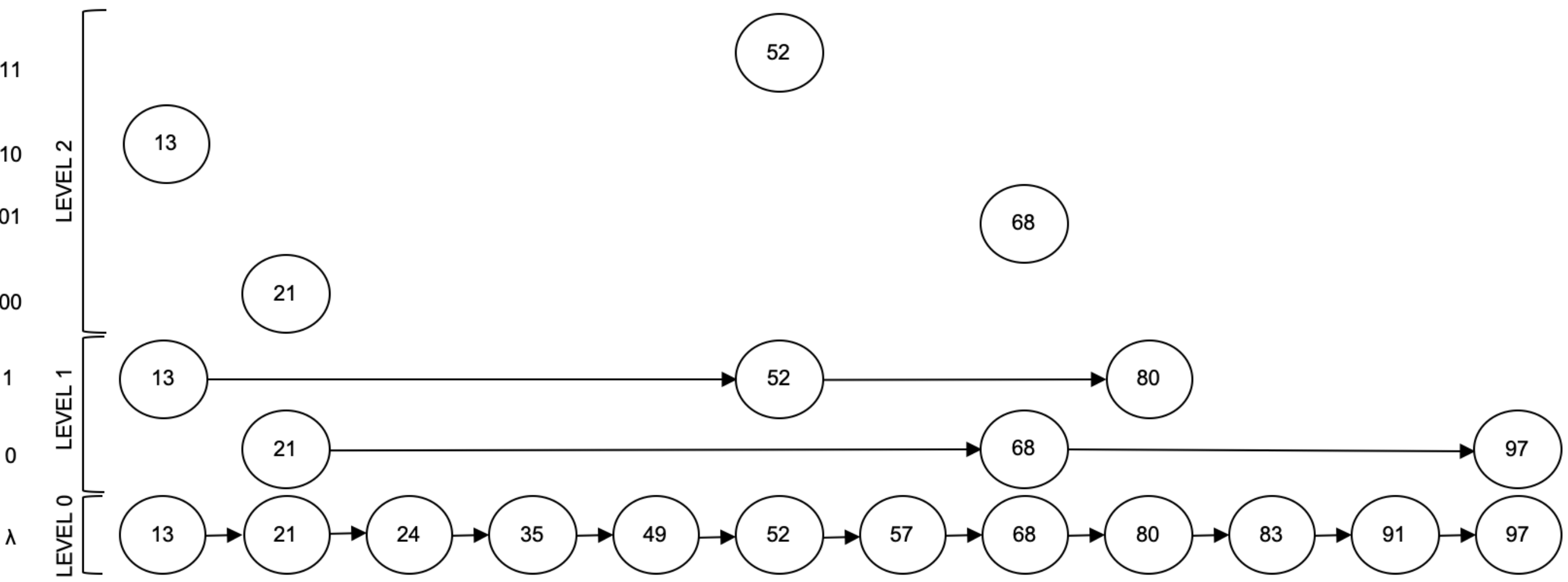}
	\caption{A sparse skip graph: elements are in level $i$ \emph{of any skip list} with expectation $1/{2^i}$.}
	\label{fig:sparseskipgraph}
\end{figure}

\label{sparseSkipGraphs}

\section{Related Work}
\label{Sec-RelatedWork}

\textbf{Data intensive applications.} Applications, particularly DBMSs, are not only concerned with \emph{data placement} (somewhat related to our goals), but also with \emph{task scheduling} (a non-goal here). Some approaches consider NUMA locality in the query evaluation and planning~\cite{numaDB2, numaDB3, numaDB4}, and others do at the OS level~\cite{numaDB1}. We have a simple approach to data placement: threads index their inserted elements locally and the dataset is effectively partitioned; our focus is on the \emph{data access pattern}. Systems such as~\cite{phoenix} are particularly concerned with data access pattern issues, as their operations are mostly uniform among threads. We do not discuss task scheduling, as we \emph{do not} rely on approaches such as delegation~\cite{delegation,qdlocking} or flat-combining~\cite{flat-combining,combining}, which \emph{do} bring concerns related to this issue because of the non-uniform work division among threads.


\textbf{Skip Lists and Skip Graphs.} Skip lists first appeared in~\cite{skiplists-concpugh}, although~\cite{skiplist-optimistic, skiplist-lazy, skiplist-lockfree} were most widely discussed in the literature~\cite{TAOMP}. Skip lists have also been used to implement priority queues, either exact~\cite{pqe,pqlotanshavit,pqsundelltsigas} or with a relaxed definition~\cite{pqrelaxed3,pqrelaxed4, pqrelaxed2,pqrelaxed1}. SkipNets~\cite{skipnets} are similar (if not identical) to skip graphs, proposed relatively at the same time. We consider those equivalent, and equally applicable. Skip graph variations, such as in~\cite{skipgraphs-rainbow}, typically address issues related to distributed systems, such as node size; we are aware of a single concurrent implementation in shared memory in~\cite{skipgraphs-concurrent}, although it is lock-based, in contrast with both of our lock-free variants. Our implementation relies heavily on laziness, as we postpone much of the internal work until they are absolutely needed. The ``No Hotspot'' skip list \cite{nohotspot} uses similar lazy principles, albeit with a different protocol. The ``Rotating'' skip list has a novel construction (``wheels'') meant to improve cache efficiency and locality, and also constitutes a modern, state-of-the-art implementation.

\textbf{NUMA awareness and layered design.} The work presented in~\cite{numaLock} gives a systematic approach to provide NUMA-awareness to locks. Tailor-made data structures for NUMA systems, such as~\cite{cphash, numastack} have also been developed, using (now) standard techniques such as elimination~\cite{elimination} and delegation~\cite{delegation}. We think that ``blackbox'' approaches, such as in~\cite{blackbox}, are interesting as they relieve systems programmers from ``customizing'' their data structures for NUMA, a notoriously complicated task for non-specialists (and specialists alike~\cite{TAOMP}). 
NUMASK \cite{numask} is an interesting skip list that uses its higher levels as a hierarchical ``index'' to the bottom-level list, which stores the dataset. In our case, the dataset is located in a structure of its own, a multi-level skip graph. This allows for our data partitioning mechanism, designed to (i) reduce non-local NUMA traffic, and particularly avoid traversals that navigate back and forth across NUMA nodes; and (ii) reduce contention by creating areas within the shared structure where only subsets of threads operate. Our thread-local indexing, similarly, is more detached from the dataset, and could be implemented with any sequential, navigable map. In our implementation, for example, our map is actually a combination of a search tree and a hash table. Finally, our indexes are not replicated, but partitioned. Even with our optional load balancing mechanism in place, threads donate nodes but do not replicate their indexing. Apart from differences of granularity and function, the idea of separating thread-local views and shared views has been seen before in \cite{frontaldatastructure2}, although their approach is more akin to combining, as they eventually merge thread-local views into the shared structure from time to time.

\section{Implementation Details}
\label{Sec-ImplementationDetails}

Section~\ref{Sec-ArchitectureOverviewBackground} gives a design overview, and this section gives more insight into implementation and correctness, and discusses optimizations related to laziness and physical removal. We will discuss first the implementation of map ADTs using our technique, and later highlight the differences in the relaxed priority queue protocols.

\textbf{General implementation concepts.} Our skip graph has a \lstinline{head} array pointing to the first shared node in each shared linked list. Each shared node \lstinline{s} has a reference array to successors in each level-$i$ list, denoted \lstinline{s.next[i]}, and also called the \emph{level-$i$ reference} of \lstinline{s}. Each reference has a \lstinline{marked} and a \lstinline{valid} bit used for lazy removal, and the functions \lstinline{s.getMark(i)}, \lstinline{getValid(i)}, \lstinline{s.casMark(i, exp, new)}, and \lstinline{s.casValid(i, exp, new)} operate on \lstinline{s.next[i]}. The functions \lstinline{s.getMarkValid(i)} and \lstinline{s.casMarkValid(i, (mExp, vExp), (mNew, vNew))} operate on \lstinline{marked} and \lstinline{valid} bits of \lstinline{s.next[i]} simultaneously.
The mapping performed by local structures is done through maps that manage local nodes. We can ask a local node $l$ for its key with \lstinline{l.getKey()}, and for its value with \lstinline{l.getValue()}. The latter always returns a shared node \lstinline{s} with the same key as the local node $l$. Traversals in the local structure are done with \lstinline{l.getPrev()} or \lstinline{l.getNext()}. The local structure for thread $i$, a C++ \lstinline{std::map}, is denoted by \lstinline{localstructures[i]}. Each has an auxiliary hash table \lstinline{hashtable[i]}, allowing threads to consult a fast hashtable \cite{RobinHood} before consulting a slower map. Therefore, our local structures, in practice, are implemented with two complementary, sequential data structures.

\textbf{Laziness.} State-of-the-art concurrent data structures rely heavily on postponing internal work until they become absolutely needed, in the hope that they become unnecessary. Our lazy layered skip graph employs this principle as: (i) Insertions are done in the level-0 linked list first, and only complete the insertion when the node in question is requested to be in the top level in order to start a search operation. (ii) Removals are performed by ``invalidating'' a \emph{linked} node, and nodes are marked for physical removal when the threads that originally inserted those nodes find them invalidated after a minimal \emph{commission period} for which they exist in the structure. Experimentally (Sec.~\ref{Sec-Evaluation}), we found that a commission period proportional to the number of threads, say $350000 \cdot T$ cycles, for instance, performs \emph{very well} under high-contention without introducing too much overhead in low-contention (in the latter, a longer commission period could leave the data structure much larger at times). An alternative to our commission period policy would be having a background thread that started the physical removal of invalid nodes. However, in order to increase NUMA locality, and to preserve the benefits from our data partitioning technique, we want the traversing threads themselves to start this process, using a \emph{local} synchronization operation. In that context, the commission period becomes a distributed control that avoids nodes to be permanently allocated. (iii) We perform physical removal using our relink optimization (discussed below), but only when \emph{substituting} a chain of marked references with an inserting node. Although this protocol has the potential to leave too many unmarked nodes in the data structure, we verified experimentally that the actual number of traversed shared nodes per operation is less than in a skip list up to 96 threads.
\label{commission}

\textbf{Physical removal.} When a node's level-0 reference \lstinline{marked} (resp. \lstinline{valid}) bit is set, we say the shared node is \emph{marked} (resp. \emph{invalid}). The concepts of \emph{unmarked} and \emph{valid} are obvious. In the non-lazy version, we do not use the \lstinline{valid} bit, and in that case an unmarked shared node indicates presence in the abstract key set, while a marked shared node indicates a \emph{logically deleted} node.
In the lazy version, an unmarked, valid node indicates presence in the abstract set; an unmarked, invalid node indicates absence from the abstract set (logically deleted), but also that the process of physically unlinking the node has not started; a marked node can only be invalid, and the process of physically unlinking is ready to start. Each shared node \lstinline{s} has a field \lstinline{s.allocTimestamp}, set at shared node construction, used to calculate when a node's commission period has elapsed, thus making it a candidate for physical removal (see the discussion on laziness, above). The removal process happens in distinct phases. (i) A thread that traverses the data structure will mark a shared node with an expired commission period only if such node has been inserted by the thread in question. (ii) When threads remove marked nodes from their local structures, they also mark upper-level references in order to promote physical removal. As in many cases logically deleted nodes have not been inserted in upper levels, upper-level reference marking is not always necessary or complete. (iii) Finally, traversing threads perform the physical cleanup using the relink optimization discussed below.

In textbook skip lists \cite{TAOMP}, we indicate willingness to physically remove a node \lstinline{s} by marking its \lstinline{s.next[i]} references for all levels $i \ldots 0$ the node belongs to. Within a level $i$, searches performed on behalf of insertions and removals physically remove nodes with marked \lstinline{next[i]} references by employing a \emph{single} CAS per node. In contrast, both our skip lists or skip graphs remove \emph{sequences} of marked references with a single CAS per sequence, a trivial optimization that we we denote by \emph{relink optimization}. The correctness of this protocol is trivial when we consider that marked references are immutable. Algorithmic details are discussed in Appendix~\ref{App-Contains}.
\label{physicalRemoval}

\textbf{Insertion.} Algorithms~\ref{alg:insert}, \ref{alg:insertHelper}, and~\ref{alg:lazyInsert} show the insert operation for \lstinline{key}. In \lstinline{insert()} (Alg.~\ref{alg:insert}), if the thread finds a reference to a shared node with the goal key (call it \lstinline{result}), it invokes \lstinline{insertHelper()} (Alg.~\ref{alg:insertHelper}), which atomically checks \lstinline{result}'s \lstinline{marked} and \lstinline{valid} bits and finishes the operation if: (I-i) \lstinline{result} was seen as an unmarked valid node, so we linearize a failed insertion right before our check (Alg.~\ref{alg:insertHelper}, line~\ref{alg:insertHelper:linDuplicate}). (I-ii) \lstinline{result} was atomically changed from invalid to valid, so we linearize the successful insertion right at that time (Alg.~\ref{alg:insertHelper}, line~\ref{alg:insertHelper:flipped}). If \lstinline{insertHelper()} cannot finish the operation (by returning false), \lstinline{result} got marked, so (i) we clean the thread's local structure (Alg.~\ref{alg:insertHelper}, lines~\ref{alg:insertHelper:cleanup1} and~\ref{alg:insertHelper:cleanup2}); and (ii) we call \lstinline{lazyInsert()} (Alg.~\ref{alg:lazyInsert}), which will complete the insertion.
\begin{algorithm}[htb]
\caption{bool Layered::insert(K key, V value)}
\label{alg:insert}
\begin{algorithmic}[1]
	\State SharedNode result = hashtables[threadId].find(key)

	\If{result $\ne$ \textbf{null}}
		\State bool returnValue = false
		\If{SG::insertHelper(result, returnValue)} \label{alg:insert:callHelper}
			\State \Return returnValue
		\EndIf
	\EndIf

	\State SharedNode insertedSkipNode = \textbf{null}
	\If{SG::lazyInsert(key, insertedSkipNode)}
		\State localstructures[threadId].insert(key, insertedSkipNode)
		\State hashtables[threadId].insert(insertedSkipNode)
		\State \Return \textbf{true}
	\EndIf

	\State \Return \textbf{false}
\end{algorithmic}
\end{algorithm}

\begin{algorithm}[htb]
\caption{bool SG::insertHelper(SharedNode toInsert, \textbf{ref} bool returnValue)}
\label{alg:insertHelper}
\begin{algorithmic}[1]
	\While{\textbf{true}}
		\If{\textbf{not} toInsert.getMark(0)}
			\If{toInsert.getMarkValid(0) == (false, valid)} \Comment{Duplicate} \label{alg:insertHelper:linDuplicate}
				\State bool returnValue = \textbf{false}
				\State \Return \textbf{true}
			\EndIf
			\If{toInsert.casMarkValid((false, invalid), (false, valid))} \Comment{Flipped \lstinline{valid}} \label{alg:insertHelper:flipped}
				\State returnValue = \textbf{true}
				\State \Return \textbf{true}
			\EndIf
		\Else
			\State localstructures[threadId].erase(toInsert.getKey()) \label{alg:insertHelper:cleanup1}
			\State hashtables[threadId].erase(toInsert.getKey()) \label{alg:insertHelper:cleanup2}

			\State \textbf{break}
		\EndIf
	\EndWhile

	\State \Return \textbf{false}
\end{algorithmic}
\end{algorithm}

The procedure \lstinline{lazyInsert()} (Alg.~\ref{alg:lazyInsert}) starts by calling \lstinline{getStart()} (Alg.~\ref{alg:getStart}), which finds the closest preceding, unmarked shared node pointed by the local structure (line~\ref{alg:lazyInsert:getStart}). We call this node \lstinline{currentStart}. We then perform a search in line~\ref{alg:lazyInsert:lazyRelinkSearch}, starting from \lstinline{currentStart}, by calling \lstinline{lazyRelinkSearch()} (Alg.~\ref{alg:lazyInsert:lazyRelinkSearch}, page \pageref{searchProcedures}). After this search: (I-iii) If an unmarked shared node with the goal key is found (call it \lstinline{successors[0]}), a call to \lstinline{insertHelper()} (line~\ref{alg:lazyInsert:insertHelper}) will define our linearization, just like our previous discussion of Alg.~\ref{alg:insert}. If that call returns false, \lstinline{successors[0]} became marked, so we retry the search at line~\ref{alg:lazyInsert:retry}. (I-iv) If an unmarked node with the goal key is not found, we try to insert the node on line~\ref{alg:lazyInsert:CAS}. If the attempt fails, we retry the search until either (I-iv-a) we link the new shared node in line~\ref{alg:lazyInsert:CAS}, so we linearize at the point the CAS succeeded, noting that shared nodes are allocated as unmarked and valid; or (I-iv-b) find another unmarked shared node with the goal key, reverting to case (I-i). The lazy insertion only happens at level 0. In line~\ref{alg:lazyInsert:updateStart}, \lstinline{updateStart} (Alg.~\ref{alg:updateStart}, Appendix~\ref{App-ExtraInsertionAlgorithms}) makes sure that \lstinline{currentStart} is unmarked, otherwise it traverses the local structure backwards and updates \lstinline{currentStart} to the closest unmarked shared node seen. 
\begin{algorithm}[htb]
\caption{bool SG::lazyInsert(T key, \textbf{ref} SharedNode insertedSkipNode)}
\label{alg:lazyInsert}
\begin{algorithmic}[1]
	\State Array predecessors[MaxLevel], middle[MaxLevel], successors[MaxLevel]

	\CommentLine{Shared node being inserted in the bottom level only}
	\State SharedNode toInsert

	\CommentLine{Search the local structure and get the closest starting point}
	\State SharedNode currentStart = getStart(key) \label{alg:lazyInsert:getStart}

	\While{\textbf{true}}
		\If{SG::lazyRelinkSearch(key, predecessors, middle, successors, currentStart)} \label{alg:lazyInsert:lazyRelinkSearch}
			\State returnValue = false
			\If{SG::insertHelper(successors[0], returnValue)} \label{alg:lazyInsert:insertHelper}
				\State \Return returnValue
			\Else
				\State \textbf{continue} \label{alg:lazyInsert:retry}
			\EndIf
		\EndIf

		\State toInsert.setNext(0, successors[0])

		\If{\textbf{not} predecessors[0].casNext(middle[0], toInsert)} \label{alg:lazyInsert:CAS}
			\State updateStart(currentStart) \label{alg:lazyInsert:updateStart}
			\State \textbf{continue}
		\EndIf

		\State insertedSkipNode = toInsert

		\State \Return \textbf{true}
	\EndWhile
\end{algorithmic}
\end{algorithm}

The procedure \lstinline{getStart()} (Alg.~\ref{alg:getStart}) finds an unmarked shared node pointed by the local structure that closest precede \lstinline{key}. It performs an initial search at line~\ref{alg:getStart:initialSearch}, then traverses the tree backwards (line~\ref{alg:getStart:traverseBackwards}) for as long as the shared nodes pointed to by the local structure are found marked. Along this traversal, if we find a suitable candidate, but one whose insertion in the upper levels of its shared skip list has not been completed, we call \lstinline{finishInsert()} (Alg.~\ref{alg:finishInsert}, Appendix~\ref{App-ExtraInsertionAlgorithms}). This is done at line~\ref{alg:getStart:finishInsert} to \emph{try} to complete the node's insertion at upper levels. This call can fail if the node gets marked before all levels are linked, and in that case we continue our traversal. In line~\ref{alg:getStart:finishInsert}, the call to \lstinline{updateStart()} (Alg.~\ref{alg:updateStart}, Appendix~\ref{App-ExtraInsertionAlgorithms}) finds a suitable starting point for \lstinline{finishInsert()}, but does not finish insertions itself. So, \lstinline{updateStart()} is a simplified of \lstinline{getStart()} that does not finish insertions lazily, and ignores nodes that are not fully inserted.
\begin{algorithm}[htb]
\caption{LocalStructureIterator LocalStructure::getStart(K key)}
\label{alg:getStart}
\begin{algorithmic}[1]
	\State iterator = localstructures[threadId].getMaxLowerEqual(key) \label{alg:getStart:initialSearch}

	\While{iterator $\ne$ \textbf{null}} \label{alg:getStart:lazylocalremoval}
		\State SharedNode sharedNode = iterator.sharedNode
		\If{\textbf{not} sharedNode.getMark(0) \textbf{or} \textbf{not} sharedNode.getMark(MaxLevel)}
			\If{\textbf{not} sharedNode.inserted}
				\If{SG::finishInsert(sharedNode, updateStart(iterator))} \label{alg:getStart:finishInsert}
					\State \Return iterator \Comment{Node has just been fully inserted}
				\Else
					\CommentLine{Erase below does not invalidate the iterator}
					\State localstructures[threadId].erase(key)
					\State hashtables[threadId].erase(key)
				\EndIf
			\Else
				\State \Return iterator \Comment{Node already found fully inserted}
			\EndIf
		\Else
			\CommentLine{Erase below does not invalidate the iterator}
			\State localstructures[threadId].erase(key)
			\State hashtables[threadId].erase(key)
		\EndIf
		\State iterator = iterator.getPrev() \label{alg:getStart:traverseBackwards}
	\EndWhile

	\State \Return iterator
\end{algorithmic}
\end{algorithm}

\textbf{Removal and Contains.} Removals are similar to insertions, and are performed with three analogous algorithms. We present those in detail in Appendix~\ref{App-Contains}, and we give here only an overview. The procedure \lstinline{remove()} (Alg.~\ref{alg:remove}, Appendix~\ref{App-Remove}) uses the local hashtable to locate a shared node with the goal key. If such node is found, it then tries to linearize an operation by calling \lstinline{removeHelper()} (Alg.~\ref{alg:removeHelper}, Appendix~\ref{App-Remove}). This algorithm is similar to \lstinline{insertHelper} (Alg.~\ref{alg:insertHelper}) for the insertion. When \lstinline{removeHelper()} is called over a node, call it \lstinline{result}, it may either: (i) See \lstinline{result} as an unmarked, invalid node, so we linearize a failed removal right at that time. There cannot exist another unmarked, valid node in this case, which justifies our linearizability. (ii) Successfully unset the valid bit of \lstinline{result}, so we linearize the successful removal at that time as well. If we fail to complete either of (i) or (ii), the algorithm \lstinline{lazyRemove()} (Alg.~\ref{alg:lazyRemove}, Appendix~\ref{App-Remove}) will complete the removal (case (iii)). This algorithm essentially uses a search procedure, \lstinline{retireSearch()} (Alg.~\ref{alg:retireSearch}, Appendix~\ref{App-Contains}) and: (iii-a) If it succeeds locating an unmarked node with the goal key, \lstinline{removeHelper()} (line~\ref{alg:lazyRemove:removeHelper}) will define our linearization, falling back to (i) or (ii); or (iii-b) If \lstinline{lazyRemoval()} fails locating an unmarked node with the goal key, we linearize a failed removal at the time we found the element succeeding that key in the bottom level.
Contains is similar yet simpler than insertions. We present the algorithms in detail in Appendix~\ref{App-Contains}, and give here only an overview for the sake of space. The algorithm \lstinline{contains()} (Alg.~\ref{alg:contains}, Appendix~\ref{App-Contains}) first tries to locate an element through the local hashtable. If an unmarked, valid node is found, we linearize a successful contains at this point (case (i)). Otherwise, we call \lstinline{SG::contains()} (Alg.~\ref{alg:sgContains}, Appendix~\ref{App-Contains}), which will complete the procedure. This will invoke \lstinline{retireSearch()} (Alg.~\ref{alg:retireSearch}, Appendix~\ref{App-Contains}), trying to locate an unmarked shared node with the goal key. If one is found: (ii) if we see such node unmarked and valid, we linearize a successful contains at that time. (iii) Otherwise, we linearize a failed contains at the earliest of (i) the time we failed to see the node was unmarked and valid; or (ii) right after the node became marked.

\textbf{Search.} We use two search procedures. We show \lstinline{lazyRelinkSearch()} in Alg.~\ref{alg:lazyRelinkSearch} below, and \lstinline{retireSearch()} in Alg.~\ref{alg:retireSearch} in Appendix~\ref{App-Contains}. Consider the insertion of a new node \lstinline{toInsert}. The procedure \lstinline{lazyRelinkSearch()} identifies, at all levels: (A) which nodes should precede \lstinline{toInsert} (referenced in \lstinline{predecessors}); (B) which nodes should succeed \lstinline{toInsert} (referenced in \lstinline{successors}); and (C) which nodes, referenced in the \lstinline{middle} array, had \lstinline{predecessors[i].getNext(i)} right at the moment each predecessor \lstinline{predecessor[i]} was identified. Between \lstinline{predecessors[i]} and \lstinline{successors[i]} we have a sequence of nodes with their level-$i$ references marked. We hope these nodes get replaced, \emph{through a single CAS operation}, with \lstinline{toInsert} (hence implementing our relink optimization protocol, p. \pageref{physicalRemoval}). Importantly, during this process, line~\ref{alg:lazyRelinkSearch:checkRetire} uses \lstinline{checkRetire()} (Alg.~\ref{alg:checkRetire}, Appendix~\ref{App-Contains}) to check if nodes are invalid and their commission period has expired, in which case these nodes get marked. Remove and contains operations use \lstinline{retireSearch()}, a slightly quicker variant that does not keep track of \lstinline{predecessors}, \lstinline{successors}, and \lstinline{middle}.
\label{searchProcedures}
\begin{algorithm}[htb]
\caption{bool SG::lazyRelinkSearch(K key, SharedNode[] predecessors,\\ SharedNode[] middle, SharedNode[] successors, LocalStructureIterator currentStart)}
\label{alg:lazyRelinkSearch}
\begin{algorithmic}[1]
	\State current = \textbf{null}

	\For{level = MaxLevel $\rightarrow$ 0}
		\State current = originalCurrent = previous.getNext(level)

		\While{current.getMark(0) \textbf{or} current.checkRetire())} \label{alg:lazyRelinkSearch:checkRetire}
			\State current = current.getNext(level)
		\EndWhile

		\While{current.getKey() < key}
			\State previous = current
			\State current = originalCurrent = previous.getNext(level)

			\While{current.getMark(0) \textbf{or} current.checkRetire())}
				\State current = current.getNext(level)
			\EndWhile
		\EndWhile

		\State predecessors[level] = previous
		\State middle[level] = originalCurrent
		\State successors[level] = current
	\EndFor

	\State \Return (successors[0].getKey() == key \textbf{and not} successors[0].getMark(0))
\end{algorithmic}
\end{algorithm}

\textbf{Unbalanced workloads.} We provide an optional mechanism for handling unbalanced workloads, and we address the following scenarios: (i) Some threads may only insert, while others only perform removals/contains. (ii) Distinct groups of threads may insert in distinct partitions of the element space. Both scenarios are problematic because threads do not necessarily find good starting points for their search operations if their local structures are empty or skewed towards a partition of the element space.
Our load-balancing mechanism is based on having threads donate a fraction of their nodes inserted in the shared structure so they are added to local structures of other threads. A background thread takes into account the number of inserted elements announced by every worker thread, and, based on those numbers, continuously indicates to each worker thread the fraction of inserted nodes that are requested for donation. The worker threads place such fraction of inserted nodes into donation queues (one per worker thread), which are collected by the background thread and distributed uniformly among all other worker threads. Threads inspect receiving queues for incoming nodes, and add them into their local structures. Specifically, if a worker thread $T_i$ announces the insertion of $I_i$ elements out of a total of $I_T = \sum_{\set{0 \le i < T}} I_i$ elements, then if $1/T > (I_i/I_T - 1/T)$, then $T_i$ donates a fraction of $I_i/I_T$ of its elements; otherwise, it donates a fraction of $(I_i/I_T - 1/T)$ of them.
%
Some degree of non-uniformity, however, is not expected to impact the shared structure efficiency globally. If shared lists are cleaned less often, then it means these lists are not often operated by threads that share them. Hence, all other threads would be traversing or cleaning over different lists, still operating normally.
Also, note that donated nodes have been inserted in their bottom-level by a thread $T_i$, while they might be inserted in upper levels by a thread $T_j, j \ne i$. We are currently working on delegating the creation of upper levels to the original thread $T_i$, avoiding to cross NUMA nodes when building up those levels. On lazy implementations, where levels are only built when needed, the impact of this problem is reduced.
\label{discussion:homogeneity}

\textbf{Priority Queues.} Note that our layered structure can implement a priority queue ADT in addition to sets/maps. Like many alternative priority queues in the literature \cite{pqlotanshavit,pqsundelltsigas,pqe}, we mark elements in the bottom level of the skip graph in order to logically delete elements. However, our layered approach gives the opportunity to perform physical cleanup on local structures using \emph{combining}~\cite{flat-combining,combining}: because we always remove a \emph{whole prefix} of the local structure containing logically deleted nodes, many nodes could be removed at once, at cost comparable to a single removal. Besides, our skip graph seems to be suitable to implement \emph{relaxed priority queues} \cite{pqrelaxed3,pqrelaxed4,pqrelaxed2}. We see two immediate approaches: (i) the application of the spraying technique of \cite{pqrelaxed1} over skip graphs; and (ii) a custom protocol that traverses the skip graph deterministically, marking elements along this traversal.

Regarding option (i), the main idea of \cite{pqrelaxed1} is to disperse threads over the skip list bottom level through a random walk called a \emph{spray}. If we apply the technique to a skip graph, each thread $T_i$ would navigate only through its associated skip list $L_i$. We see several advantages of performing spray operations over skip graphs rather than skip lists. Our partitioning scheme will still incur in better memory locality and reduced contention, as discussed before. On the other hand, we show that spraying over skip graphs has a slightly bigger removal range (same asymptotically, but higher nevertheless), as seen in Theorem~\ref{thr:range_spray_sg}.

Regarding option (ii), we implement a deterministic traversal in the skip graph, marking elements along the way. Informally, a thread $T_i$ starts at the highest level of its associated skip list, traverses marked nodes, and attempts to mark at the current position. If the attempt succeeds, a node has been logically deleted, otherwise the thread moves down a level, traverses marked nodes, and proceeds similarly. At level 0, two mark attempts are tried, and upon failing the second one, the process is restarted. This process is shown in Alg.~\ref{alg:sg_remove_min} in Appendix~\ref{App-AnalysisSpray}. We prove, on Theorems~\ref{thr:cas_slspray} and~\ref{thr:cas_sgspray}, that we have at most 2 threads contending for the same node along the traversal. In other words, the number of CAS operations required to logically delete a sequence of $T$ nodes in our deterministic protocol is $2T$.

\section{Analysis of Spray Operations on Perfect Skip Graphs}
\label{App-AnalysisSpray}

In this section, we provide formal arguments related to the performance of spray operations as in \cite{pqrelaxed1} as applied over skip graphs. Part of the analysis presented here closely follows the aforementioned paper, while others are particular to skip graphs.
Recall that $T$ is the number of threads in the system. We will assume that $\log T$ is integral, and that for the number of nodes $N$ in our skip graphs or lists, $N \geq T$ holds. When convenient we will explicitly assume $T = 2^n$, where $n \ge 1$. The particular variant of skip graphs we refer to, throughout our analysis, is the one where in each level $i$, where $0 \leq i \leq \log T - 1$, there are $2^i$ lists as described in Sec.~\ref{Sec-ArchitectureOverviewBackground}.

\subsection{Reach of Skip Graph Spraying}

\begin{definition}
A skip graph is \emph{perfect} if the absolute value of the difference in the positions of any two consecutive nodes on level $i$ is $2^i$, for all $0 \leq i \leq \log T - 1$.
\end{definition}

The operation $\tt{SPRAY(H,L,D)_j}$ defines a random walk starting at height $\tt{H}$, i.e., at the head of a list $j$ at level $\tt{H}$ of the skip graph. The walk continues forward to the right a uniformly random number of consecutive nodes chosen from $\tt{[0,L]}$ then moves down the skip graph $\tt{D}$ levels. This is repeated until we reach the bottom list and (possibly) walk forward to the right one last time. Thus the walk ends at position $x$ in the bottom list. We are interested in characterizing the probability that any $\tt{SPRAY}$ operation starting from any $j$ of the $2^{\log T -1}$ lists of the maximum level of the skip graph lands on position $x$ in the bottom list.

\begin{definition}
Let $j$ be any of the $2^{\log T -1}$ lists of the maximum level of a perfect skip graph, and let $x$ be the position of any node in the perfect skip graph. Then $\tt{F(x)_j}$ denotes the probability that for fixed $\tt{H}$, $\tt{L}$, and $\tt{D}$, $\tt{SPRAY(H,L,D)_j}$ lands on position $x$ in the bottom list.
\end{definition}

In the proof of the theorem below, we closely follow \cite{pqrelaxed1} with appropriate changes and modifications for the perfect skip graph. We assume w.l.o.g. that $\log T$ is odd.

\begin{theorem}
\label{thr:sgrange}
Fix $\tt{H} = \log T -1 $, $\tt{L} = \log T$, and  $\tt{D} = 1$. Then for any position $x$, $\tt{F(x)_j} \leq \frac{1}{T}$.
\end{theorem}
\begin{proof}
Let $j$ be any of the $2^{\log T -1}$ lists of the maximum level of a given skip graph, and let $x$ be the position of any node in the perfect skip graph. Let $pos_j$ denote the position of the head of list $j$. If $x < pos_j$, then trivially $\tt{F(x)_j} = 0$. Hence, we assume w.l.o.g. $x \geq pos_j$. Furthermore let $\bar{x} := x - pos_j$. We simplify the exposition and show that $\tt{F(x)_j} \leq \frac{1}{T}$ by making our arguments on $\bar{x}$.
Note that in order for $\tt{SPRAY(\log T -1,\log T,1)_j}$ to land on position $x$, the following must hold:
\begin{equation}
    \sum_{i =0}^{\log T -1} a_i2^i = \bar{x}, \label{sum_of_parts}
\end{equation}
where for each $i$, $a_i$ is chosen independently from $[0,\log T]$. Thus, we are interested in the probability that the sum of the different $a_i$, each chosen independently from $[0,\log T]$ yield Equation~\eqref{sum_of_parts}.
Instead of calculating this probability directly we will give an upper bound for it based on a simple argument with respect to the parity of the bits in the binary expansion of $\bar{x}$.

The idea is to give the probability that each randomly chosen $a_i$, starting with $i=0$, contributes to the correct parity of the binary expansion of $\bar{x}$ moving from left to right. This is possible because of the shifting factor of $2^i$ for each $a_i$. Certainly, all randomly chosen $a_i$ that satisfy Equation~\eqref{sum_of_parts} must satisfy this property, but not the other way around. Therefore, this probability yields an upper bound for $\tt{F(x)_j}$. We show this more precisely as follows.

Let $\bar{x}(l)$ denote the $l$-th least significant bit of $\bar{x}$, and for each $i$, let $a_i(l)$ denote the $l$-th least significant bit of $a_i$. In order to ensure the correct parity of the binary expansion of $\bar{x}$ moving from left to right by \emph{sequentially} choosing $a_i$ in $[0,\log T]$ starting with $i=0$, it is easy to see that $a_0(0) = \bar{x}(0)$ must hold. More generally, because of the shifting factor of $2^k$ for each $a_k$, the following equation (whose left hand side represents bit addition)
\begin{equation}
a_k(0) + a_{k-1}(1) + \ldots + a_0(k) + c \equiv \bar{x}(k)\bmod 2, \label{gen_parity}  
\end{equation}
must hold for all $1\leq k \leq \log T -1$, where $c$ is determined by the previous choice of $a_0 , a_1, \ldots a_{k-1}$.

We will derive the desired upper bound by considering the probability of sequentially and randomly picking $a_i$ that respect the parity of the binary expansion of $\bar{x}$ as explained above, starting with $a_0$, then continuing with $a_1$ and so on until $a_{\log T - 1}$. We say that a chosen $a_k$ is a  $\tt{match}$, if it respects the parity of the binary expansion of $\bar{x}$ as in Equation~\eqref{gen_parity}. Thus we are interested in the probability on the right hand side of the following inequality:
\begin{equation}
\tt{F(x)_j} \leq \tt{Pr}[a_0 \text{ is a } \tt{match}] \prod_{i=1}^{\log T -1}Pr[a_i \text{ is a } \tt{match} | A_{i-1} ], \label{prob_ineq}
\end{equation}
where $A_{i-i}$ is the event where all $a_k$ such that $0\leq k\leq i - 1$ are a $\tt{match}$. First we will explicitly derive the right hand side of Inequality~\eqref{prob_ineq} then we will give a simple argument as to why the inequality holds.

First of all, it is clear that $\tt{Pr}[a_0 \text{ is a } \tt{match}] = \frac{1}{2}$. This is because we assumed $\log T$ is odd and exactly half the numbers in $[0,\log T]$ are even or odd. Furthermore, this can be readily generalized for $\tt{Pr}[a_i \text{ is a } \tt{match} | \tt{A_{i-1}}]$, where $1\leq i \leq \log T -1$. Given $A_{i -1}$, i.e., $a_0,a_1, \ldots a_{i-1}$ that have respected the parity of $\bar{x}$ as above, exactly half the possible choices in $[0,\log T]$ will contribute an $a_i$ such $a_k(0)$ satisfies Equation (2), where $k = i$. Hence $\tt{Pr}[a_i \text{ is a } \tt{match} | \tt{A_{i-1}}] = \frac{1}{2}$, for each $1\leq i \leq \log T -1$. Thus,
\begin{equation}
   \tt{Pr}[a_0 \text{ is a } \tt{match}] \prod_{i=1}^{\log T -1}Pr[a_i \text{ is a } \tt{match} | A_{i-1} ] = \frac{1}{2^{\log T}} = \frac{1}{T}.\label{RHS_prob}
\end{equation}
Consider choosing $a_k$ sequentially as before, but this time the chosen $a_k$ in addition to satisfying Equation (2), must also represent an actual possible part of a $\tt{SPRAY}$ path that satisfies Equation~\eqref{sum_of_parts}. We say that a chosen $a_k$ is an $\tt{exact}$ $\tt{ match}$ if this is the case. Then, it follows that
\begin{equation*}
    \tt{F(x)_j} = \tt{Pr}[a_0 \text{ is an } \tt{exact}\text{ }\tt{match}] \prod_{i=1}^{\log T -1}Pr[a_i \text{ is an } \tt{exact}\text{ } \tt{match} | A_{i-1} ],\label{exact_prob}
\end{equation*}
where $A_{i-i}$ is the event where all $a_k$ such that $0\leq k\leq i - 1$ are an $\tt{exact}$ $\tt{match}$. Hence, Inequality~\eqref{prob_ineq} holds, and the desired result follows from Equation~\eqref{RHS_prob}.
\end{proof}
\begin{theorem}\label{thr:range_spray_sg}
For each $\tt{SPRAY(\log T - 1,\log T, 1)_j}$, for any list $j$ in the maximum level of a perfect skip graph, the position of the node on which the operation lands is at most $\frac{T}{2} + \log T\cdot(T -1) -1$.
\end{theorem}
\begin{proof}
Consider a perfect skip graph such that $N \geq \frac{T}{2} + T\log T$. We show that any $\tt{SPRAY(\log T - 1,\log T, 1)_j}$ operation can reach from its starting position is $\frac{T}{2} + \log T\cdot(T-1)T$ units forward to the right. For any list $j$ in the maximum level of a perfect skip graph, if the path chosen by a $\tt{SPRAY(\log T - 1,\log T, 1)_j}$ operation moves forward to the right by $\log T$ nodes on the corresponding lists in every possible level, we arrive at a total of $\sum_{i=0}^{\log T -1}2^i \log T = \log T\sum_{i=0}^{\log T -1}2^i = \log T\cdot(T-1)$ units forward to the right.

By Proposition~\ref{fistTpos}, if we label the positions of the heads of lists in the maximum level moving from left to right in order of appearance, we arrive at  $pos_{\log T -1}^0, pos_{\log T -1}^1, \ldots pos_{\log T -1}^{2^{\log T -1} -1}$, where $pos_{\log T -1}^i = i$, for all $0 \leq i \leq 2^{\log T -1}-1$.  Consider list $k$ whose head is the node in position $pos_{\log T -1}^{2^{\log T -1}-1} = 2^{\log T -1}-1 = \frac{T}{2} -1$. Then, by the argument in the previous paragraph, there is a $\tt{SPRAY(\log T - 1,\log T, 1)_k}$ operation that lands in position $\frac{T}{2} + \log T\cdot(T -1) -1$.
\end{proof}

We denote our protocol using the $\tt{SPRAY(\log T - 1,\log T, 1)_j}$ operations on perfect skip graphs as $\tt{SPRAY\_SG}$, and on perfect skip lists as $\tt{SPRAY\_SL}$. Finally we denote our deterministic, mark-along protocol on perfect skip graphs as $\tt{SGMARK}$. The protocol is shown in Alg.~\ref{alg:sg_remove_min}, below.

\begin{algorithm}[htb]
	\caption{bool PQ::removeMin(K key)}
	\label{alg:sg_remove_min}
	\begin{algorithmic}[1]
		\State $r = \ $ random number [0, $t - 1$] where $t$ is the number of threads
		\If{$r = \ $ threadId}
			\CommentLine{Clean as described in \cite{pqrelaxed1}}
		\EndIf
		\While{\textbf{true}}
			\State level = MaxLevel
			\While{level $\ge$ 0}
				\State node $=$ next unmarked node

				\If{node \textbf{is} tail}
					\State level = level - 1
					\State \textbf{continue}
				\EndIf

				\If{node.mark()}
					\State \Return \textbf{true}
				\EndIf

				\State level = level - 1
			\EndWhile

			\While{performed $\le 2$ times}
				\State node $=$ next unmarked node
				\If{node \textbf{is} tail}
					\State \Return \textbf{false}
				\EndIf
				\If{node.mark()}
					\State \Return \textbf{true}
				\EndIf
			\EndWhile

		\EndWhile
		\State \Return \textbf{false}
\end{algorithmic}
\end{algorithm}

We conclude that any $\tt{SPRAY}$ operation in  $\tt{SPRAY\_SG}$, using the same technique from \cite{pqrelaxed1}, satisfies the same probability bounds as any $\tt{SPRAY}$ operation with the same parameters in $\tt{SPRAY\_SL}$. The difference is that the range of furthest reaching $\tt{SPRAY}$ operation is greater in $\tt{SPRAY\_SG}$ than  $\tt{SPRAY\_SL}$. In particular for $\tt{SPRAY}$ operations with parameters as in Proposition~\ref{thr:range_spray_sg} the range for $\tt{SPRAY\_SG}$ is $\frac{T}{2} + T\cdot \log T$, whereas the range for $\tt{SPRAY\_SL}$ is $T\cdot \log T$. The range of $\tt{SGMARK}$ is exactly $T$. 
In the next section, our approach will consist on defining a skip list spray that disperses threads among the same range as $\tt{SGMARK}$, and then indicate that contention is expected to be smaller in a skip graph due to \emph{structural features of the skip graph itself}.

\subsection{Skip List Sprays over $T$ Elements}

\begin{definition}
A skip list is \emph{perfect} if the absolute value of the difference in the positions of any two consecutive nodes on level $i$ is $2^i$, and every list on level $i$ contains at least $\frac{T}{2^i}$ nodes, for all $0\leq i \leq \log T -1$.
\end{definition}

The operation $\tt{SPRAY(H,L,D)}$ defines a random walk starting at height $\tt{H}$, i.e., at the head of the list at level $\tt{H}$ of the skip list. The walk continues forward to the right a uniformly random number of nodes chosen from $\tt{[0,L]}$ then moves down the skip list $\tt{D}$ levels. This is repeated until we reach the bottom list and (possibly) walk forward to the right one last time. Thus the walk ends in position $x$ of the bottom list.

\begin{definition}
Let $x$ be the position of any node in a perfect skip list. Then $\tt{F_p(x)}$ denotes the probability that, for some fixed $\tt{H}$, $\tt{L}$, and $\tt{D}$, $\tt{SPRAY(H,L,D)}$ lands at position $x$ in the bottom list.
\end{definition}

\begin{proposition}\label{reachT}
Let $x$ be the position of any of the first $T$ nodes of a perfect skip list. There is a $\tt{SPRAY(\log T -1,1,1)}$ that lands at position $x$ in the bottom list.
\end{proposition}
\begin{proof}
The furthest $\tt{SPRAY}$ operation moves forward to the right a total of $T-1$ units as follows, $\sum_{i = 0}^{\log T -1} 2^i = 2^{\log T} - 1 = T -1$. Hence we reach the $T$-th node. Suppose there is some $\tt{SPRAY}$ operation that reaches a node whose position is $k$, where $0< k < T -1$. Consider the following two cases:
\begin{itemize}
    \item The path determined by the given $\tt{SPRAY}$ operation consists of a final step forward. By ignoring the final step forward we arrive at a path that reaches position $k-1$.
    \item The path determined by the given $\tt{SPRAY}$ operation does not consist of a final step forward. Consider the smallest level $i$ where the path turns left (moving from bottom to top). The preceding node $j$ on the list in level $i$ is in position $k - 2^i$. Moving a level down from $j$ then forward to the right every possible level we cover a total of $\sum_{l=0}^{i-1}2^n = 2^i -1$ units horizontally and reach position $k-1$.
\end{itemize}
\end{proof}
\begin{corollary}
Fix $\tt{H} = \log T -1 $, $\tt{L} = 1$, and  $\tt{D} = 1$. Let $x$ be the position of any of the first $T$ nodes of a perfect skip list. Then $\tt{F_p(x)} = \frac{1}{\tt{T}}.$
\end{corollary}
\begin{proof}
From Proposition~\ref{reachT} we know that there exist some $\tt{SPRAY(\log T -1,1,1)}$ operation that reaches any of the first $T$ nodes, and from its proof we see that no $\tt{SPRAY}$ operation can reach beyond the first $T$ nodes. Given any position $x$, where $0 \leq x \leq T-1$, next we show that the $\tt{SPRAY}$ operation that lands on position $x$ is path-wise unique.

Fix any arbitrary $x$, where $0 \leq x \leq T-1$, and consider some $\tt{SPRAY(\log T -1,1,1)}$ operation that lands on the node in position $x$. Suppose there exists another path-wise distinct $\tt{SPRAY}$ operation that lands on the node in position $x$, which we denote by $\tt{SPRAY}_2$. Consider the largest level $i$ and node $j$ where the paths of the two path-wise distinct $\tt{SPRAY}$ operations differ. Let the position of node $j$  be $k$. W.l.o.g. assume the path which belongs to the first $\tt{SPRAY}$ operation moves forward to the right from node $j$ on level $i$ before descending a level. This ensures that $k + 2^i \leq x$. The path induced by the $\tt{SPRAY}_2$ operation does not move forward to the right and instead directly descends a level from node $j$. This implies the path can move forward to the right from position $k$ at most $\sum_{n=0}^{i-1} 2^n = 2^i -1$. Thus the furthest position a path induced by $\tt{SPRAY}_2$ can reach is $k + 2^i -1 < x$, which is a contradiction.  

Since there is only one unique path that reaches each $x$, at each level $i$, where $0\leq i \leq \log T -1$, the probability that any given $\tt{SPRAY}$ generates the correct portion of the path at level $i$ is exactly $1/2$. Since this probability is independently uniform at each level $i$, we arrive at $\tt{F_p(x)} = \prod_{i=0}^{\log T -1}\frac{1}{2} = \frac{1}{2^{\log T}}= \frac{1}{\tt{T}}.$
\end{proof}
Let $X$ denote the number of $\tt{SPRAY}$ operations in a perfect skip list until every position $x$, where $0\leq x \leq T-1$, is reached. We are interested in the expected number of $\tt{SPRAY}$ operations which ensure that every position $x$ is reached at least once. To arrive at the expected number of $\tt{SPRAY}$ operations we show our question of interest is simply the well-known Coupon Collector's Problem \cite{ferrante2014coupon}. Thus the desired result is proportional to $T\log T$.

Now that we defined a skip list spray that disperses threads among a segment of size $T$, we show that the expected number of operations to mark all elements in that segment is $\Theta(T \log T)$.

\subsection{Number of Successful CAS Operations when Spraying on Skip Lists}

\begin{theorem}\label{thr:cas_slspray}
$\E[X]$ is in $\Theta(T\log T)$.
\end{theorem}
\begin{proof}
Let $X_i$ denote the number of $\tt{SPRAY}$ operations performed while exactly $i-1$ positions in the bottom list have been already reached. Then $X = \sum_{i =1}^{T}X_i$ holds. Thus, each $X_i$ is a geometric random variable. If exactly $i-1$ positions have already been reached by  previous $\tt{SPRAY}$ operations then the probability that next $\tt{SPRAY}$ operation reaches a different position is simply $p_i = 1 - \frac{i-1}{T}$. This is exactly the Coupon Collector's Problem. We outline the analysis for completion and clarity as follows:
\begin{equation*}
    \E[X]=\E[\sum_{i =1}^{T}X_i]=\sum_{i =1}^{T}\E[X_i]=\sum_{i =1}^{T}\frac{T}{T-i+1}=T\sum_{i=1}^{T}\frac{1}{i}= T\cdot H(T),
\end{equation*}
where $H(T)$ is the $T$-th harmonic number. Using elementary calculus one can show that $H(T) = \ln T + \Theta(1)$ \cite[33]{mitzenmacher2017probability} and hence $T\cdot H(T) = T\ln T + \Theta(T)$. Thus it follows that $\E[X]$ is in $\Theta(T\log T)$.
\end{proof}
\begin{corollary}
Let $T \geq 4$. Then $\E[X] \geq 2T$.
\end{corollary}
\begin{proof}
We can use well-known elementary techniques to bound $H(T)$ as follows:
\begin{eqnarray*}
H(T) &=& 1 + \frac{1}{2} + \frac{1}{3} + \ldots + \frac{1}{2^{\log T}}\\
     &=& 1 + \frac{1}{2} + \left(\frac{1}{3} + \frac{1}{4}\right) + \left(\frac{1}{5} + \frac{1}{6} + \frac{1}{7} + \frac{1}{8}\right) + \left(\frac{1}{2^{\log T -1}+1} + \ldots +\frac{1}{2^{\log T}}\right)\\
     &\geq& 1+ \frac{1}{2} + \left(\frac{1}{4} + \frac{1}{4}\right) + \left( \frac{1}{8} + \frac{1}{8} + \frac{1}{8} + \frac{1}{8}\right) + \left(\frac{1}{2^{\log T}} + \ldots +\frac{1}{2^{\log T}}\right)\\
     &\geq& 1 + \frac{1}{2}\log T
\end{eqnarray*}
Then from Theorem~\ref{thr:cas_slspray} we have that $\E[X]= T\cdot H(T) \geq T \left(1 + \frac{1}{2}\log T\right)$, which gives the desired result for all $T\geq 4$.
\end{proof}

We now characterize how a custom relaxed priority queue algorithm, designed precisely in order to exploit the central structural features of the skip graph, can remove elements from a range of $T$ elements with proven contention being exactly 2 \emph{for any number of threads $T$}.

\subsection{Number of Successful CAS Operations when Spraying on Skip Graphs with Our Algorithm}

\begin{definition}
Let $T = 2^n$, where $n\geq 1$. A perfect skip graph is \emph{minimal} if the number of nodes in the bottom list is exactly $T$.
\end{definition}
Thus each $n\geq 1$ determines a unique minimal skip graph (up to permutation of lists in each level) which we denote be  $PSG_{\tt{min}}(n)$. Given $PSG_{\tt{min}}(n)$, we will show that it contains all smaller perfect minimal skip graphs.
\begin{definition}
Consider $PSG_{\tt{min}}(n)$ and $PSG_{\tt{min}}(m)$, where $n>m\geq 1$. Then $PSG_{\tt{min}}(n)$ \emph{contains} $PSG_{\tt{min}}(m)$, if by ignoring consecutive levels and consecutive nodes of $PSG_{\tt{min}}(n)$ we arrive at $PSG_{\tt{min}}(m)$ (up to permutation of lists in each level).
\end{definition}
\begin{proposition}\label{recursive_perfect}
For all $n\geq 2$, $PSG_{\tt{min}}(n)$ contains all smaller minimal perfect skip graphs.
\end{proposition}
\begin{proof}
For $n=1$ by definition the minimal perfect skip graph is a list with $2$ nodes. Consider a list with $4$ nodes (2 copies of $PSG_{\tt{min}}(1)$). Partitioning this list into two other lists with consecutive node positions $0,2$ and $1,3$ we arrive at $PSG_{\tt{min}}(2)$. Let $n=3$. We will construct  $PSG_{\tt{min}}(3)$ from  $PSG_{\tt{min}}(2)$ in the following way. We glue two copies of $PSG_{\tt{min}}(2)$ together so that the height of the new structure is the same as $PSG_{\tt{min}}(2)$ but the bottom list has length $8$. It is straightforward to see these are the first $2$ levels of $PSG_{\tt{min}}(3)$. Thus we arrive at $PSG_{\tt{min}}(3)$ by adding its maximum level.

Suppose it holds for all $n \leq k$, for some $k>3$. Consider $PSG_{\tt{min}}(k+1)$. Ignoring the maximum level, we see by translational symmetry that the remaining structure consists of two glued copies of $PSG_{\tt{min}}(k)$. By the induction hypothesis $PSG_{\tt{min}}(k)$ contains all other smaller minimal perfect skip graphs.
\end{proof}
Let $T = 2^n$ for $n \geq 1$. Then our partition scheme ensures there are exactly two threads acting on the heads of each list of the maximum level of the perfect skip graph at the beginning of $\tt{SGMARK}$. In the following arguments we assume that the process of marking a node by one of the threads on level $n-1$ then having the other thread descend to a list on level $n-2$ occurs \emph{simultaneously} for all $2^{n-1}$ lists $k$ of level $n-1$. More generally, whenever threads are on level $i$, the ones which fail to mark nodes will \emph{simultaneously} move down a level. Furthermore we assume, that if two threads traversing the same list where the nearest unmarked node $j$ is the same for both, then the threads reach $j$ \emph{simultaneously}. Suppose two threads are acting on every list $k$ of some fixed level $i$ of a perfect skip graph, where $1\leq i \leq n -1$, such that for every list $k$ the nearest unmarked node $j$ is the same for both threads. Then we will assume all threads will \emph{simultaneously} reach their respective contested nodes.

Under these assumptions, we can model $\tt{SGMARK}$ as a game with $n$ rounds. The last round is a special case where the surviving thread will walk forward to the right one last time instead of descending. At the beginning of a regular round all surviving threads up to that round are on level $i$, where $0\leq i \leq n-2$, whereas at the end of the round some threads have moved down a level and have traversed to the right. We will see that it suffices to consider our protocol on minimal perfect skip graphs. First, we will need the following definitions, which are well-defined by Proposition~\ref{recursive_perfect}.
\begin{definition}\label{lrs_psg}
For all $n\geq 2$, we arrive at the \emph{right hand side copy} of $PSG_{\tt{min}}(n-1)$ in $PSG_{\tt{min}}(n)$, by ignoring level $n-1$ of $PSG_{\tt{min}}(n)$ and the first $2^{n-1}$ consecutive nodes of the bottom list in $PSG_{\tt{min}}(n)$.
\end{definition}
\begin{definition}\label{rhs_psg}
For all $n\geq 2$, we arrive at the \emph{left hand side copy} of $PSG_{\tt{min}}(n-1)$ in $PSG_{\tt{min}}(n)$ via translational symmetry by shifting the \emph{right hand side copy} of $PSG_{\tt{min}}(n-1)$ in $PSG_{\tt{min}}(n)$ to the left $2^{n-1}$ positions.
\end{definition}
\begin{proposition} \label{fistTpos}
Let $T=2^n$, where $n\geq 2$. Consider any level $i$ of $PSG_{\tt{min}}(n)$, where $0 \leq i\leq n-1$. Consider nodes from left to right in order of appearance and let $pos_i^0, pos_i^1, \ldots pos_i^{2^i -1}$ denote the positions of the heads of the respective $2^i$ list in level
$i$. Then $pos_i^j = j$, for all $0 \leq j \leq 2^{i}-1$.
\end{proposition}
\begin{proof}
This is trivially true for $n=2$ by inspection. Suppose it holds for some integer $k > 2$. Consider $PSG_{\tt{min}}(k+1)$. If we ignore level $k$ of $PSG_{\tt{min}}(k+1)$, we are left with the left hand side and right hand side copy of $PSG_{\tt{min}}(k)$ in $PSG_{\tt{min}}(k+1)$. Therefore, for all levels $i$ such that $0 \leq i \leq k - 1$, the statement holds by the induction hypothesis. By symmetry, the maximum level of  $PSG_{\tt{min}}(k+1)$ is partitioned into two parts, one part of the maximum level located above the left hand side copy of $PSG_{\tt{min}}(k)$ in $PSG_{\tt{min}}(k+1)$, and the other above the right hand side copy of $PSG_{\tt{min}}(k)$ in $PSG_{\tt{min}}(k+1)$. Since the maximum level of any minimal perfect skip graph contains lists with only two nodes, the portion of each list in the maximal level that is above the left hand side copy of $PSG_{\tt{min}}(k)$ in $PSG_{\tt{min}}(k+1)$ contains only one node. These are the $2^k$ heads of the lists in level $k$ of $PSG_{\tt{min}}(k+1)$, and since the maximum level is a partition of the nodes on the bottom list, after possibly reordering they have positions $0,1, \ldots 2^k -1$.
\end{proof}
\begin{proposition}\label{firstround}
For any $n\geq 2$ and $PSG_{\tt{min}}(n)$, at the end of the first round of $\tt{SGMARK}$ on $PSG_{\tt{min}}(n)$, we have $2$ threads acting on the heads of the lists of the maximum level of the right hand side copy of $PSG_{\tt{min}}(n-1)$ in $PSG_{\tt{min}}(n)$.
\end{proposition}
\begin{proof}
Fix $n \geq 2$, where $T=2^n$. According to $\tt{SGMARK}$, at the beginning of the first round there are exactly two threads acting on the heads of each of the lists in level $n-1$. By Proposition~\ref{fistTpos}, positions $pos_{{n -1}}^0, pos_{{n -1}}^1, \ldots pos_{{n -1}}^{2^{n-1} -1}$, are marked by $2^{n-1}$ threads. The remaining $2^{n-1}$ threads move simultaneously down to level $n-2$. Each list in level $n-2$ is partitioned into two lists in level $n-1$ (first two nodes in each level $n-2$ become heads of lists in level $n-1$), and it follows that for each list $k$ of level $n-2$, we have exactly two threads  acting on list $k$,  one for each of the first two consecutive nodes on that particular list. There are now a total of  $2^{n-2}$ threads  positions at the heads of all list $k$ in level $n-2$. By Proposition~\ref{fistTpos} it follows that the threads are in positions $pos_{n-2}^0, pos_{n-2}^1, \ldots pos_{n-2}^{2^{n-2}-1}$. Since there is also a thread on the second node of every list $k$ on level $n-2$, for every two threads on any such list $k$ the nearest unmarked node is the next consecutive one. Counting from the heads of the lists on level $n-2$ the third consecutive node will be $2^{n-1}$ positions away. This implies, that at the end of the first round we will have two threads on each of the positions $pos_{n-2}^{2^{n-1}}, pos_{n-2}^{2^{n-1} +1}, \ldots pos_{n-2}^{(2^{n-1} + 2^{n-2}-1)}$. These are just the heads of the lists of the maximum level of the right hand side copy of $PSG_{\tt{min}}(n-1)$ in $PSG_{\tt{min}}(n)$.
\end{proof}
\begin{corollary}\label{endstart}
Let $n\geq 2$. The end of the first round of $\tt{SGMARK}$ on $PSG_{\tt{min}}(n)$ is equivalent to the start of a first round of $\tt{SGMARK}$ on the right hand side copy of $PSG_{\tt{min}}(n-1)$ in $PSG_{\tt{min}}(n)$.
\end{corollary}
\begin{proof}
Follows directly from Proposition~\ref{firstround}.
\end{proof}
\begin{theorem} \label{thr:cas_sgspray}
Let $n\geq 2$. $\tt{SGMARK}$ on $PSG_{\tt{min}}(n)$ ensures that exactly $T=2^n$ nodes are marked. Furthermore, 2 threads contend for each of the first $T-1$ consecutive nodes, and 1 thread tries to CAS (logically delete) the last node.
\end{theorem}
\begin{proof}
Let $n=2$. This means there are $2$ threads acting on each of the nodes in position $0$ and $1$ on the second level of $PSG_{\tt{min}}(2)$. According to $\tt{SGMARK}$, one thread each will mark nodes in positions $0$ and $1$ respectively, followed by the other two threads that will descend to the bottom list. Then, both threads will move to the right simultaneously to the node in position $2$, and one of the threads will succeed in marking the node in position $2$. Finally the last thread will move to the right and mark the last node. Thus, two threads contended for each node in positions $0,1$ and $2$, and there was no contention in the last node, since at most 1 thread can try to CAS (logically delete) the last node.

Suppose the claim holds for some $k > 2$. Consider the first round of $\tt{SGMARK}$ on $PSG_{\tt{min}}(k +1)$. Following the proof of Proposition~\ref{firstround}, we see that the first $2^{k}$ nodes are marked at the start of the first round, and there are two threads contending for each of the $2^k$ nodes that are marked. By Corollary~\ref{endstart} we see that the end of the first round on $PSG_{\tt{min}}(k +1)$ is the start of the first round on the right hand side copy of $PSG_{\tt{min}}(k)$ in $PSG_{\tt{min}}(k +1)$. Therefore by the induction hypothesis $\tt{SGMARK}$ will mark exactly $2^{k} + 2^{k} = T^{k+1}$ nodes. Furthermore 2 threads contend for the first $T^{k+1}-1$ consecutive nodes, and 1 thread tries to CAS (logically delete) the last node.
\end{proof}

\section{Evaluation}
\label{Sec-Evaluation}

We conducted experiments in a system with 2 Intel Xeon Platinum 8275CL CPUs, each with 24 cores running at 3.0GHz (96 hardware threads total). The system has 192GB of memory and two NUMA nodes. The tool \texttt{numactl --hardware} reports intra-node distances of 10 and inter-node distances of 21. The system runs Ubuntu Linux 18.04 LTS with kernel 4.15.0. We use \texttt{g++ -std=c++11 -O3 -m64 -fno-strict-aliasing} to compile tests.

\textbf{Experiment setup.} We report the \emph{total} number of operations per millisecond achieved in trials having from 2 up to 96 threads. Each trial is an average of 5 runs of 10s each, and follows exactly the testing procedure of Synchrobench~\cite{Synchrobench} with the flag {\texttt{-f 1}}. This flag indicates that the testing procedure tries to match each trial's requested percent of \emph{update operations} (inserted and remove) as much as possible, and that only successful inserts or removals count as update operations. The testing procedure, as well as random number generation, are identical to Synchrobench. We run a \emph{read-heavy} (RH) load, with a requested 20\% of update operations, a \emph{write-heavy} (WH) load, with a requested 50\% of update operations, and a \emph{priority queue} (PQ) load, with 50\% of insertions and 50\% of removals, all distributed uniformly at random across all threads (except for our load-balancing tests). If X\% of operations correspond to \emph{successful} updates in a given experiment, we say we had X\% of \emph{effective updates}, and we report that percentage in each associated graphic caption. Our experiments are defined to be \emph{high contention} (HC) when the key space is $2^{8}$, \emph{medium contention} (MC) when it is $2^{11}$, and \emph{low contention} (LC) when it is $2^{17}$. The structures are preloaded with 20\% of their maximum capacity before any measurements, except for the LC tests, which are preloaded with 2.5\%. LC tests and analysis are presented in Appendix.~\ref{App-LCTests}; here we will focus on HC and MC tests. Threads are pinned to each CPU, and we fill a socket before adding threads to another socket. We allocate memory with \texttt{libnuma}, in chunks capable of holding $2^{20}$ objects, in order to amortize the expensive cost of \texttt{numa\_alloc\_local()}.

\textbf{Membership Vectors.} Our membership vectors are generated so that threads pinned to closer hardware threads in the system have their associated skip lists sharing more levels. We obtain data from \texttt{/proc/cpuinfo} on Linux, then renumber threads so the larger the absolute difference between thread identifiers $1 \ldots T$, the larger the physical distance between their associated CPUs. We consider NUMA domains, core collocation, and hardware-thread collocation in order to access distance and define our membership vectors.

\textbf{1 - Performance results.} Figures~\ref{fig:performance-hc-wh} and \ref{fig:performance-mc-wh} show write-heavy (WH) results for the HC and MC contention scenarios. Low-contention results (LC) are presented in Appendix~\ref{App-LCTests}, and read-heavy (RH) results are presented in Appendix~\ref{App-RHTests}. In our graphs, \texttt{layered\_map\_\{sg,ssg\}} refers to using C++ \texttt{std::map} in conjunction with the hash \cite{RobinHood} as local structures, respectively over regular or sparse skip graphs (p. \pageref{sparseSkipGraphs}) as shared structures; \texttt{lazy\_layered\_sg} is the lazy variant of \texttt{layered\_map\_sg}; \texttt{rotating} is \cite{rotating}, \texttt{nohotspot} is \cite{nohotspot}, and \texttt{numask} is \cite{numask} as found in Synchrobench's GitHub (mid August 2019).
For the purpose of isolating individual design components in our analysis, we also developed as \textbf{control}: a locked skip list; a concurrent skip list with the same codebase and practices as our skip graph code, including our relink optimization (p. \pageref{physicalRemoval}); a skip graph without layering; and finally layered C++ maps (i) over a linked list (\texttt{layered\_map\_ll}) and (ii) over a skip list (\texttt{layered\_map\_sl}). The former is essentially a \texttt{layered\_map\_sg} with maximum level 0, and the latter a \texttt{layered\_map\_ssg} with a single constituent skip list (hence, with no opportunity to implement our partitioning scheme). Non-layered skip lists or skip graphs have maximum level $x$ if the test's key space is of size $2^x$, and layered versions follow our partitioning scheme definitions (p. \pageref{partitioning}).

\label{discussion:performance}
\begin{figure}[htb]
    \centering
	\begin{minipage}{0.48\textwidth}
		\centering
		\includegraphics[width=\textwidth]{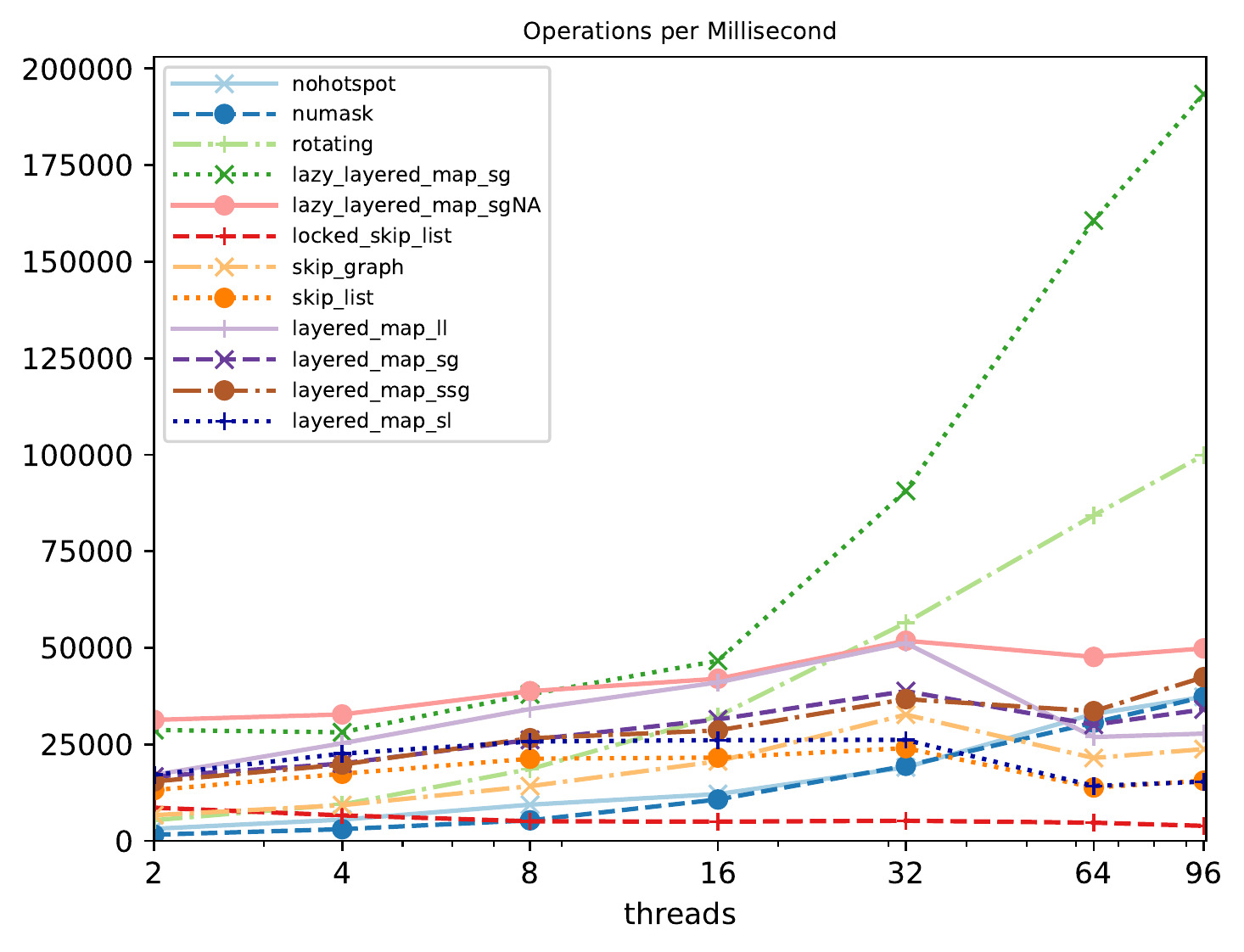}
		\caption{HC, WH: 32\% effective updates.}
		\label{fig:performance-hc-wh}
	\end{minipage}\hfill
	\begin{minipage}{0.48\textwidth}
		\centering
		\includegraphics[width=\textwidth]{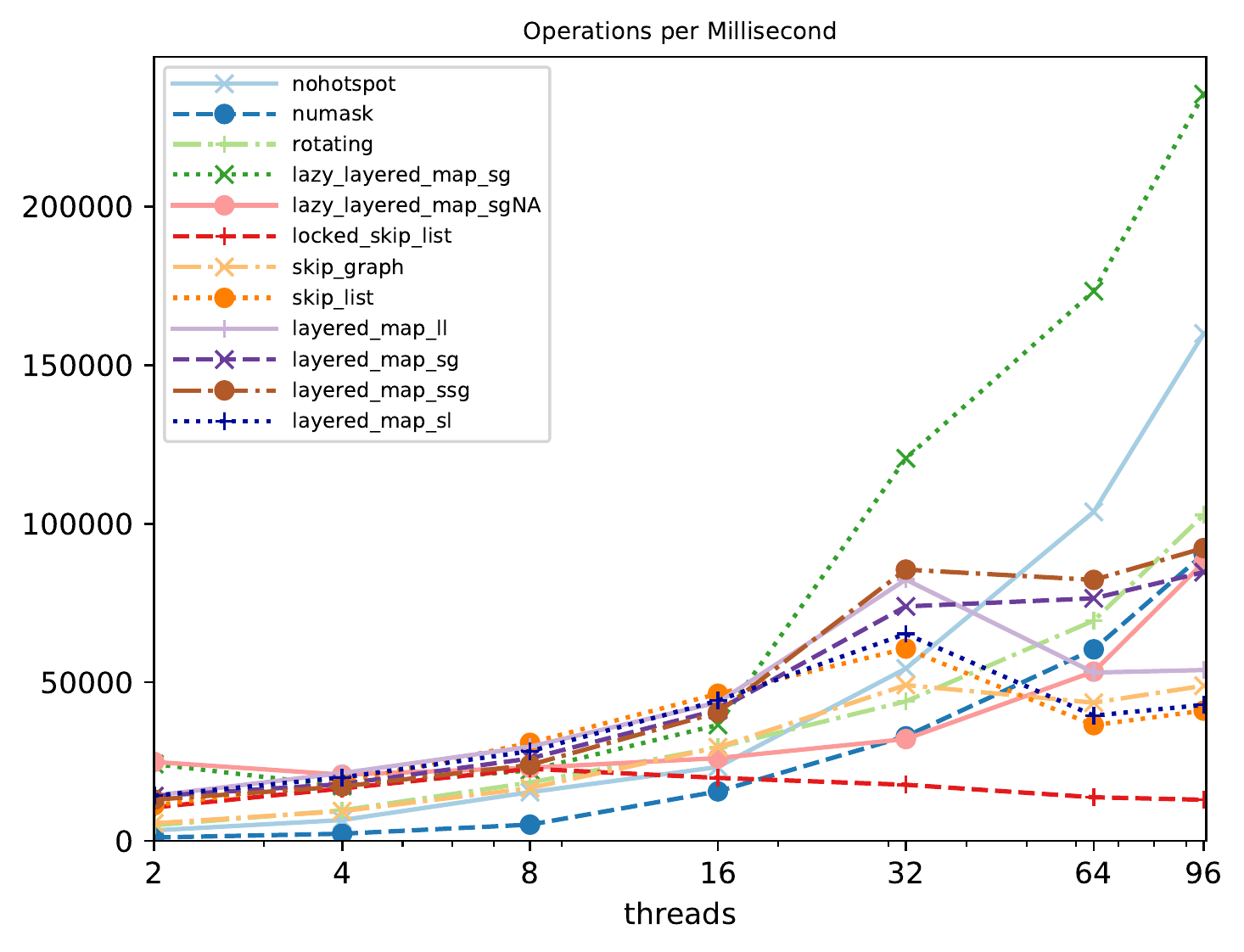}
		\caption{MC, WH: 32\% effective updates.}
		\label{fig:performance-mc-wh}
	\end{minipage}
\end{figure}

With a small key space (HC-WH), \texttt{layered\_map\_ll} performs better than \texttt{layered\_map\_sg} and \texttt{layered\_map\_sl}, but the performance degrades quickly as the key space gets bigger (MC-WH, Fig.~\ref{fig:performance-mc-wh}; LC-WH, Fig.~\ref{fig:performance-lc-wh}, Appendix~\ref{App-LCTests}). With bigger key spaces, more elements need to be traversed in the shared structure upon searches. We could be tempted to say that the multilevel shared structure in \texttt{layered\_map\_sg} is the reason it performs better in MC-WH, but note that \texttt{layered\_map\_sl} performs worse than \texttt{layered\_map\_ll} in the same case MC-WH. The \emph{reason} why \texttt{layered\_map\_sg} performs better in the MC-WH scenario seems to be, therefore, related to the partitioning scheme in the shared structure (the skip graph). Further, in the same MC-WH scenario, we note a clear performance separation between \texttt{layered\_map\_ssg} and \texttt{layered\_map\_sl} after 16 threads. The existence of multiple, overlapping of skip lists, employing a partitioning scheme across threads, is the essential difference between these two approaches.
 
As far as the lazy implementation performance, we see it as a combination of (i) the effectiveness of our partition scheme for increasing NUMA locality and reducing contention (discussed above and \emph{verified} below, in item \#2); (ii) the commission policy to unlink invalid, marked nodes (isolated right below); and (iii) the fact that with smaller key spaces, threads will more commonly find unmarked nodes through their local hashtable, which performs much better compared to the \texttt{std::map} local structure. We show a \texttt{lazy\_layered\_map\_sgNA} where we make the commission period zero, as control. Under HC-WH, \cite{rotating} performs well, and our control implementation is comparable to \cite{numask, nohotspot}. Under MC-WH, \cite{nohotspot} performs well, and our control implementation is comparable to \cite{numask, rotating}. In any case, we confirm our expectation that naive skip graphs scale poorly, because while in a skip list the expected number of levels of each node is 2, in a skip graph it is always the maximum. Further, on Tbl~\ref{tbl:metrics}, we see how \texttt{layered\_map\_sg}, without any commission period, requires a \emph{lot} more CASes per operation than other structures. With that in mind, and considering our indications that the partitioning scheme works, we have \emph{first} to make sure that skip graphs become \emph{viable} with techniques such as lazy insertions/removals, the commission period, and our relink optimizations mentioned in p.~\ref{commission}.

\textbf{2 - NUMA locality and contention reduction.} We \emph{verify} that our partitioning mechanism promotes better NUMA locality and reduced contention below. Figures~\ref{fig:heatmap-cas-lazysg} and~\ref{fig:heatmap-cas-sl} show heatmaps where coordinates $(i,j)$ indicate the distribution of CAS instructions per operation performed by thread $i$ into a node allocated by thread $j$, instrumented manually on node access functions on the 96-thread MC-WH scenario. The memory access pattern shows that the larger the distance between thread IDs (which we adjusted to match the NUMA distance), the smaller the number of memory accesses.
Comparing the lazy skip graph and a skip list (which serves as our \emph{control}, as it has been implemented using the same codebase and practices), the heatmaps indicate a \emph{dramatic} increase in CAS NUMA locality. We provide graphics for other structures (non-lazy skip graph, non-lazy sparse skip graph) in Appendix~\ref{App-LocalityHeatmaps}. Also related to locality, we expect fewer cache misses due to our partitioning mechanism, an expectation that is confirmed on Tbl.~\ref{tbl:cache_misses}.

Regarding contention, Tbl.~\ref{tbl:metrics} shows additional metrics collected via manual code instrumentation, on the 96-thread HC-WH scenario. Both our heatmaps and Tbl~\ref{tbl:metrics} \emph{do not count} CAS/read/write operations performed over an inserting node, otherwise locality would be artificially inflated with no-contention operations that are inherently local, as threads have to initialize their allocated nodes. Any CAS operation metric presented is a \emph{maintenance CAS}: an operation required to link, unlink, or cleanup nodes. 
\label{discussion:locality}
\begin{figure}[htb]
    \centering
	\begin{minipage}{0.48\textwidth}
		\centering
		\includegraphics[width=\textwidth]{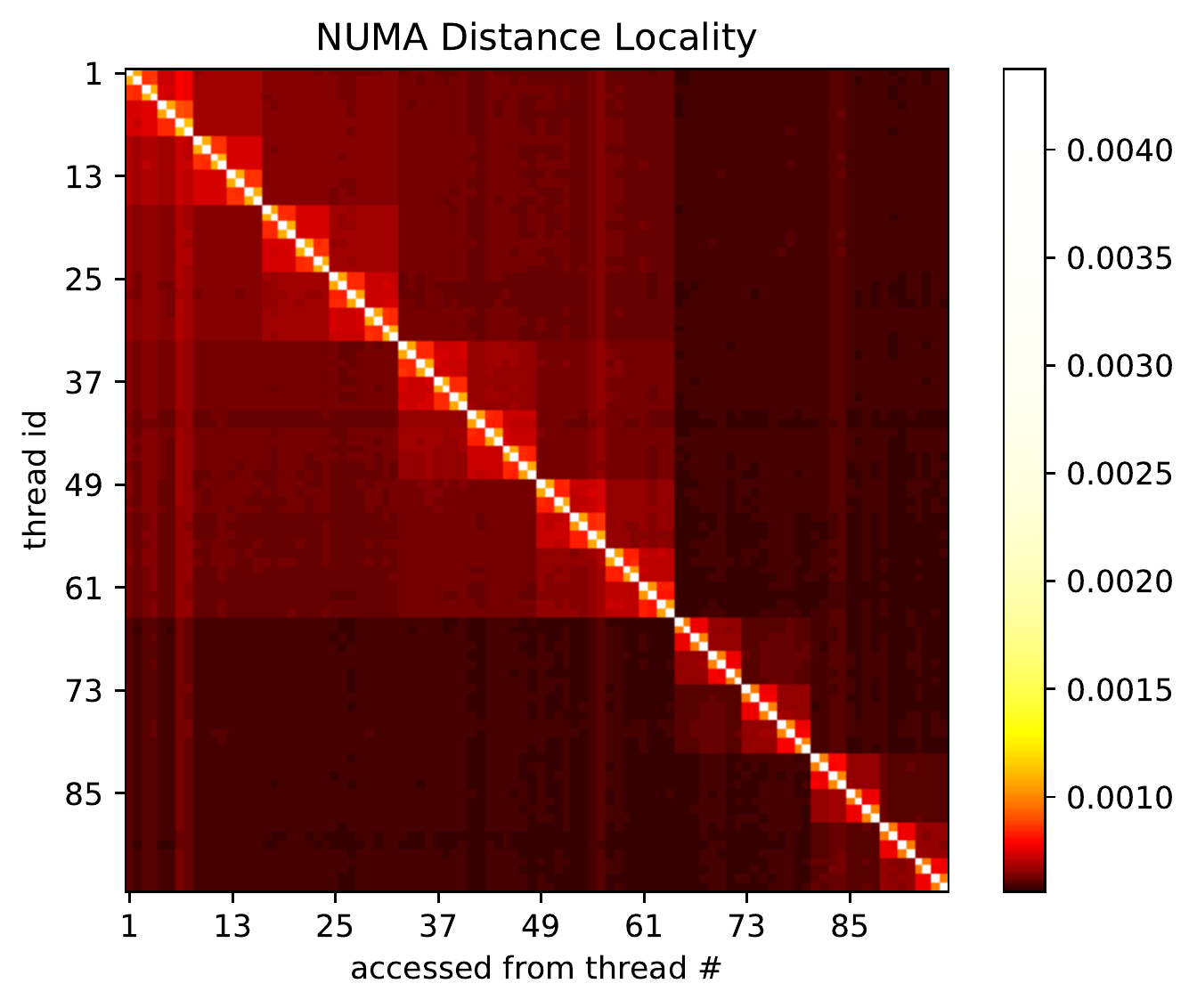}
		\caption{MC-WH CAS heatmap, lazy map/SG.}
		\label{fig:heatmap-cas-lazysg}
	\end{minipage}\hfill
	\begin{minipage}{0.48\textwidth}
		\centering
		\includegraphics[width=\textwidth]{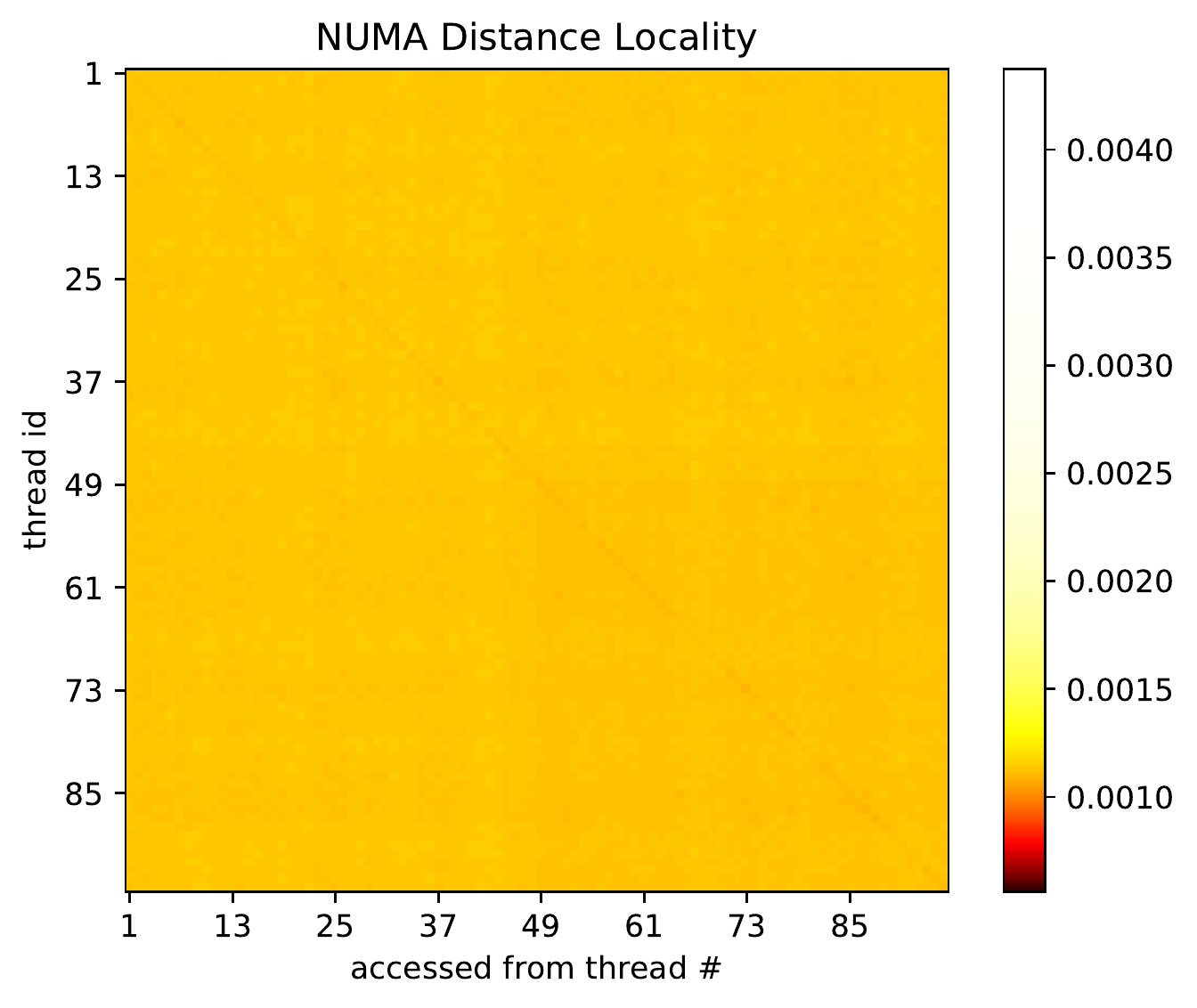}
		\caption{MC-WH CAS heatmap, skip list.}
		\label{fig:heatmap-cas-sl}
	\end{minipage}
\end{figure}

\begin{table}[htb]
	\centering
		\begin{tabular}{|c|c|c|c||c|c|c||c|c|c||c|c|c||}
		\hline
		\multirow{2}{*}{totals} & \multicolumn{3}{c}{lazy\_sg} & \multicolumn{3}{c}{map\_sg} & \multicolumn{3}{c}{map\_ssg} & \multicolumn{3}{c|}{sl}\\
		& L1 & L2 & L3 & L1 & L2 & L3 & L1 & L2 & L3 & L1 & L2 & L3 \\
		\hline
		\hline
8 & 52.8 & 12.6 & 3.1 & 56.5 & 12.4 & 3.1 & 57.6 & 12.5 & 2.7 & 65.0 & 15.1 & 3.1 \\
16 & 53.6 & 14.5 & 3.4 & 55.9 & 13.7 & 3.0 & 59.0 & 14.1 & 3.1 & 73.1 & 16.7 & 3.5 \\
32 & 87.5 & 18.7 & 3.5 & 73.9 & 14.5 & 3.0 & 93.7 & 18.1 & 2.9 & 93.0 & 24.6 & 3.7 \\
		\hline
\end{tabular}
\caption{Average (instructions \& data) cache misses \emph{per operation}, HC-WH, 96 threads. Numbers collected with PAPI.}
\label{tbl:cache_misses}
\end{table}

\begin{table}[htb]
	\centering
	\begin{tabular}{|c|c|c|c|c|}
	\hline
	& lazy map/sg & map/sg & map/ssg & skip list \\
	\hline
	local reads/op & 9.105 & 8.933 & 4.264 & 0.477 \\
	remote reads/op & 48.076 & 54.521 & 65.123 & 45.392 \\
	local maintenance CAS/op & 0.02508 & 0.177 & 0.0137 & 0.012 \\
	remote maintenance CAS/op & 0.3493 & 2.524 & 0.888 & 1.113 \\
	CAS success rate & 0.999 & 0.986 & 0.982 & 0.883 \\
	\hline
	\end{tabular}
\caption{96 threads, HC-WH. CAS/op does not include uncontended CAS operations upon node insertion. Comparing the lazy map/sg with the skip list, we can observe ~6x more CAS locality, 65\% less CAS/operation, and substantially better CAS success rate.}
\label{tbl:metrics}
\end{table}

Note that although \texttt{lazy\_layered\_map\_sg} performs slightly more reads per operation than skip lists, it performs a 68\% less remote CASes per operation (maintenance CASes). The CAS success rate is substantially higher (99\% in \texttt{lazy\_layered\_map\_sg} vs. 88.3\% in the skip list). Both the increase in NUMA locality, discussed before, and the contention reduction, just discussed, are attributed to our \emph{partitioning scheme}, designed precisely with those goals (p.~\pageref{partitioning}).

\textbf{3 - Relaxed priority queues.} Figure~\ref{fig:performance-hc-rh} tests multiple implementations for relaxed priority queues using skip graphs. The \texttt{spray} implementation consists on the application of the spraying technique of \cite{pqrelaxed1} over skip graphs; the \texttt{sg\_spray} implementation is our custom protocol that traverses the skip graph deterministically, marking elements along this traversal (Alg.~\ref{alg:sg_remove_min}, Appendix~\ref{App-AnalysisSpray}). For both approaches, we also show lazy versions (\texttt{spray\_lazy} and \texttt{sg\_spray\_lazy}). We also have a \emph{control} skip list (implemented using the same codebase and practices) where we perform spray operations as in \cite{pqrelaxed1}.

Lazy versions are expected to perform faster for similar reasons they perform better in maps, although, for relaxed priority queues, we cannot make guarantees about their degree of relaxation. The reason is that upper-level lists are substantially more sparse in our lazy implementations, and elements are added to those lists by demand, something which we cannot (at this point) reasonably model. Now focusing on the non-lazy versions, we note that \texttt{spray} scales better than \texttt{sg\_spray}. Although Theorems~\ref{thr:cas_slspray} and~\ref{thr:cas_sgspray} indicate that \texttt{sg\_spray} is subject to smaller contention, we also show there that the range of \texttt{spray} is slightly larger. Fig.~\ref{fig:keydist-hc} shows an experiment similar to the one in \cite{pqrelaxed1}, where we perform traversals and only note which nodes would be marked, without actually marking any element. The experiment shows that \texttt{sg\_spray} is indeed less relaxed than \texttt{spray}. In conjunction, Figures~\ref{fig:performance-hc} and~\ref{fig:keydist-hc} essentially exhibit a tradeoff between priority queue relaxation and scalability.

\begin{figure}[H]
    \centering
	\begin{minipage}{0.48\textwidth}
		\centering
		\includegraphics[width=\textwidth]{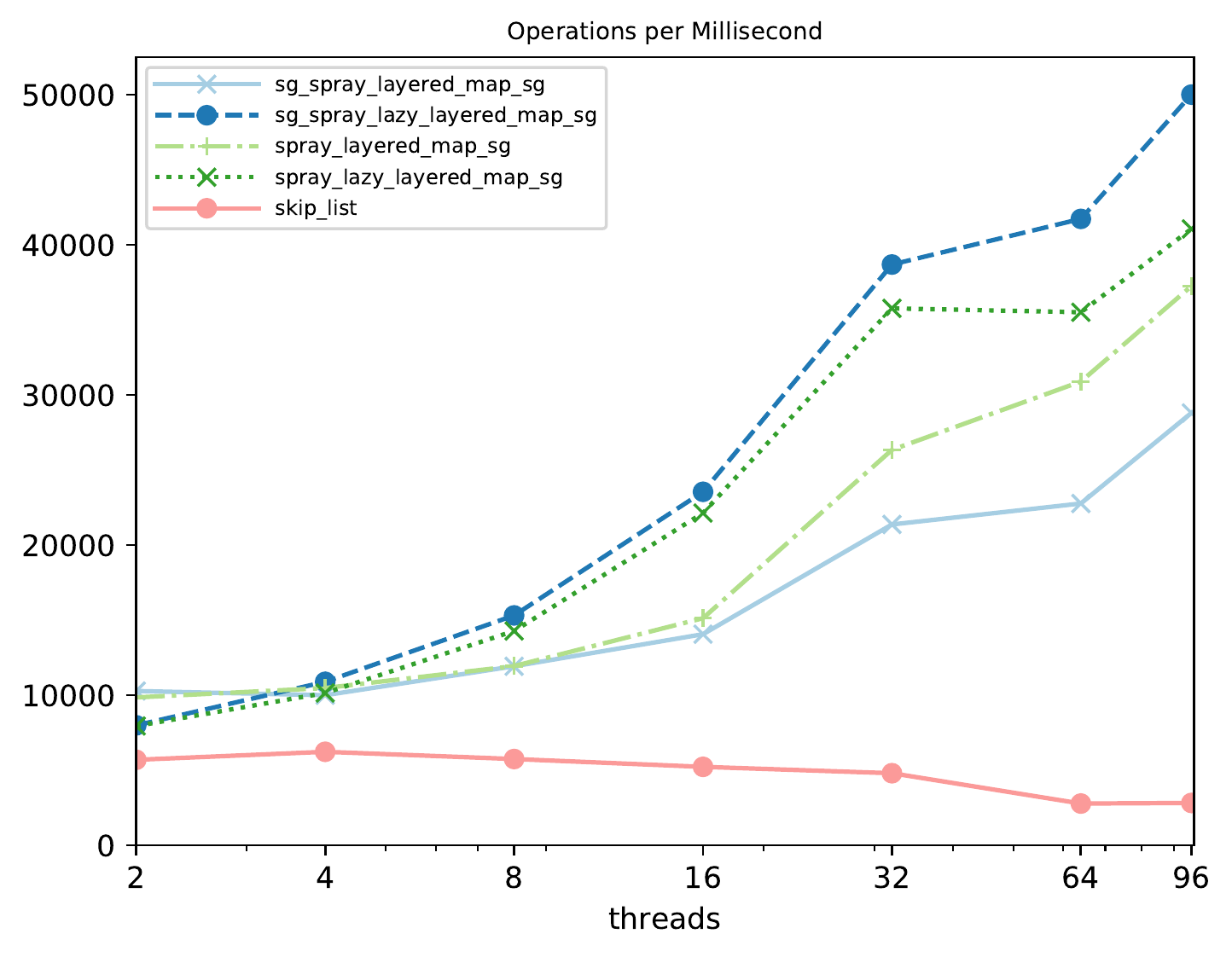}
		\caption{MC-PQ: 82\% effective updates.}
		\label{fig:performance-hc}
	\end{minipage}\hfill
	\begin{minipage}{0.48\textwidth}
		\centering
		\includegraphics[width=\textwidth]{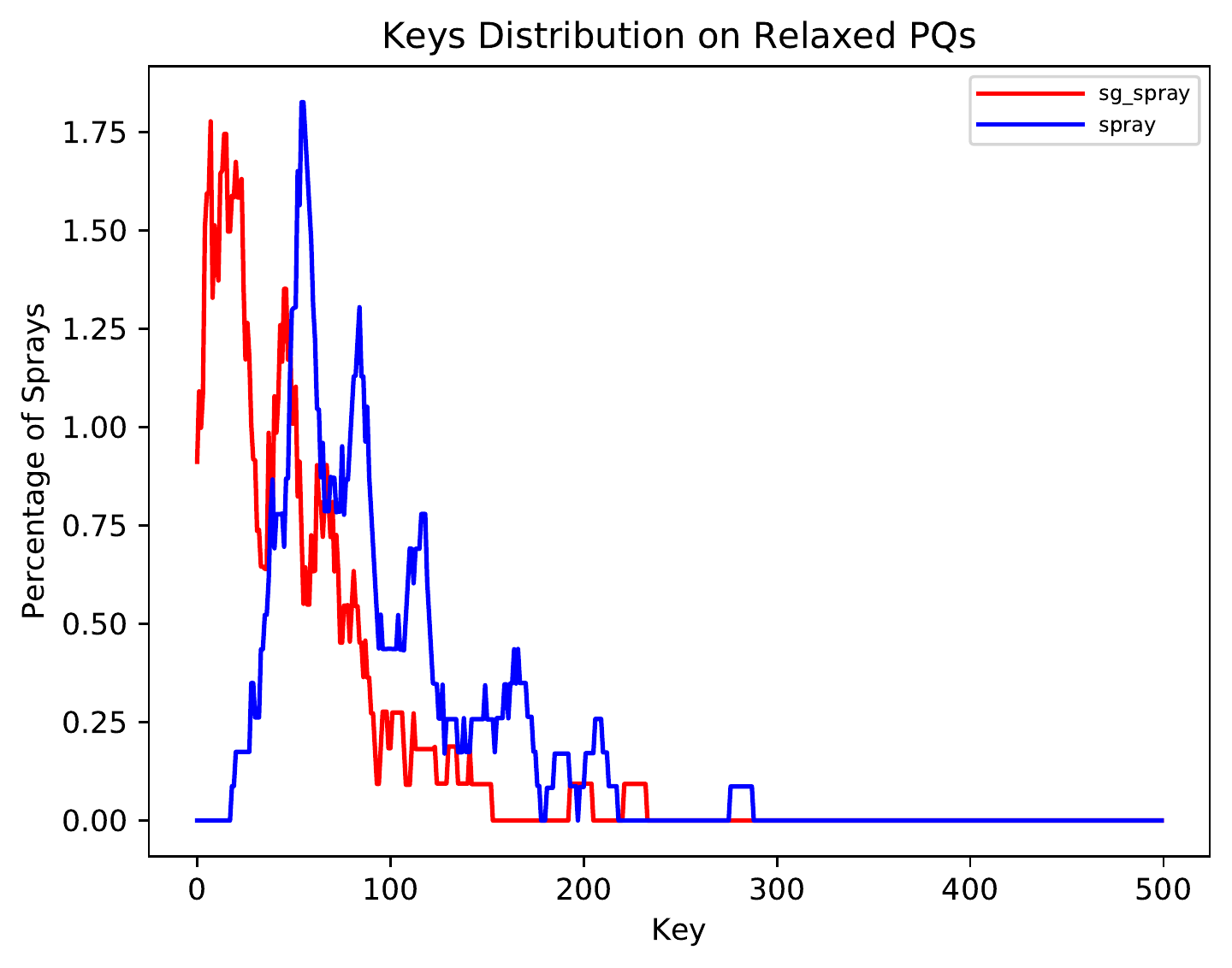}
		\caption{MC-PQ: Key distribution metric with faux removals.}
		\label{fig:keydist-hc}
	\end{minipage}
\end{figure}

\textbf{4 - Load balancing.} In order to evaluate load balancing, we consider two experiments. When only thread $T_0$ inserts, and all others perform removals and contains operations in an MC-WH experiment, Figure~\ref{fig:keysizes-mc} still shows that the \emph{sizes} of the local structures are similar. In that figure, the X-axis represent 100 discrete time points, varying from 0-10s, the Y-axis represents 8 different threads, and the color indicates the local structure size (not counting marked elements). A different experiment considers the scenario where $1/2$ of the threads insert only on the 0-1023 range of a 2048 keyspace, and $1/2$ of the threads insert only on the 1024-2047 range of the keyspace. Figure~\ref{fig:keydist-mc} shows that the key distribution of two threads that belong to two different groups are spread throughout the \emph{whole} element space. In the figure, the X-axis represents possible keys, the Y-axis represents 10 discrete time points, one per second, and the color represents whether a particular key belongs to thread 1 or 2 in their local structure.

\begin{figure}[htb]
    \centering
	\begin{minipage}{0.48\textwidth}
		\centering
		\includegraphics[width=\textwidth]{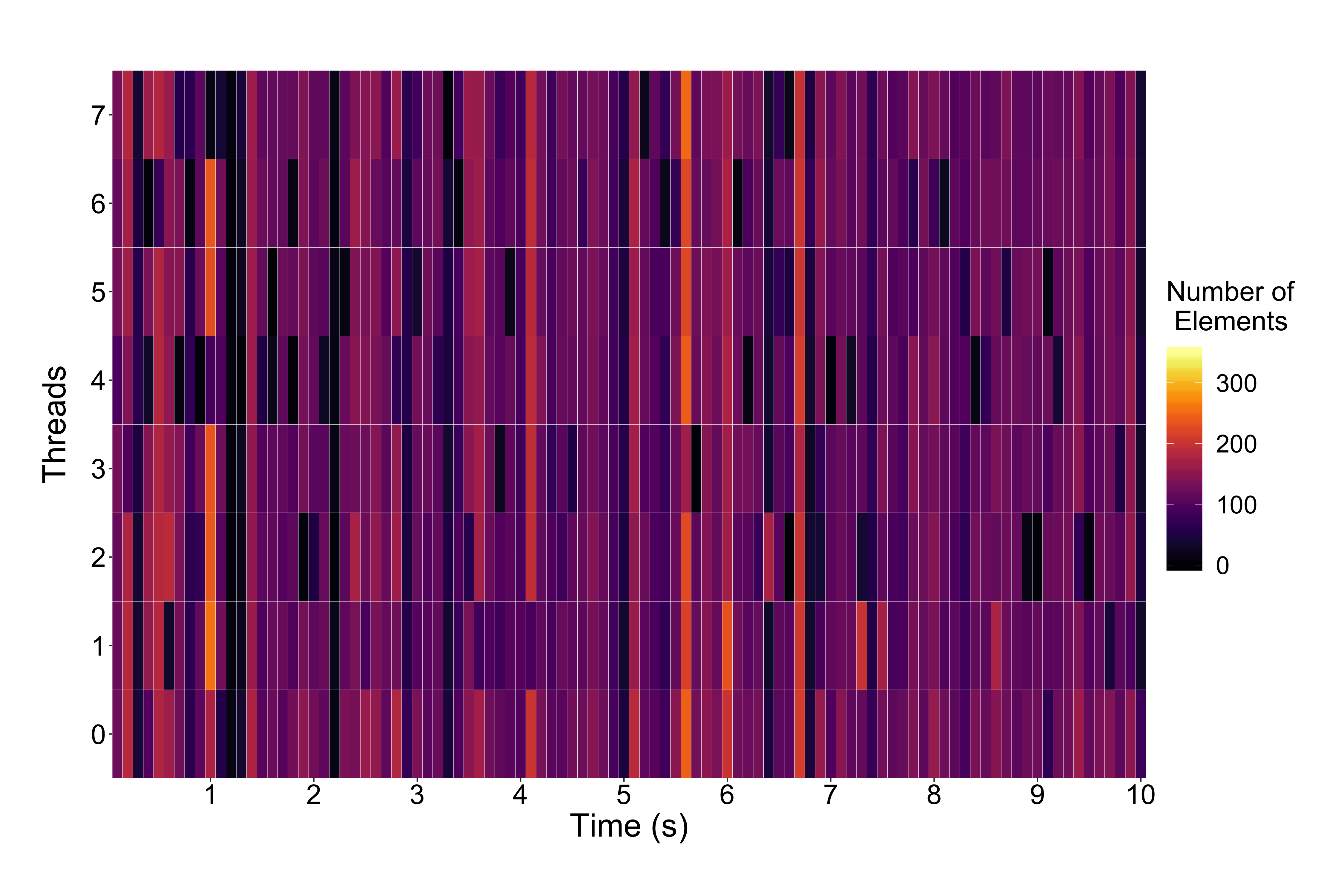}
		\caption{Local structure sizes, keyspace = 2048.}
		\label{fig:keysizes-mc}
	\end{minipage}\hfill
	\begin{minipage}{0.48\textwidth}
		\centering
		\includegraphics[width=\textwidth]{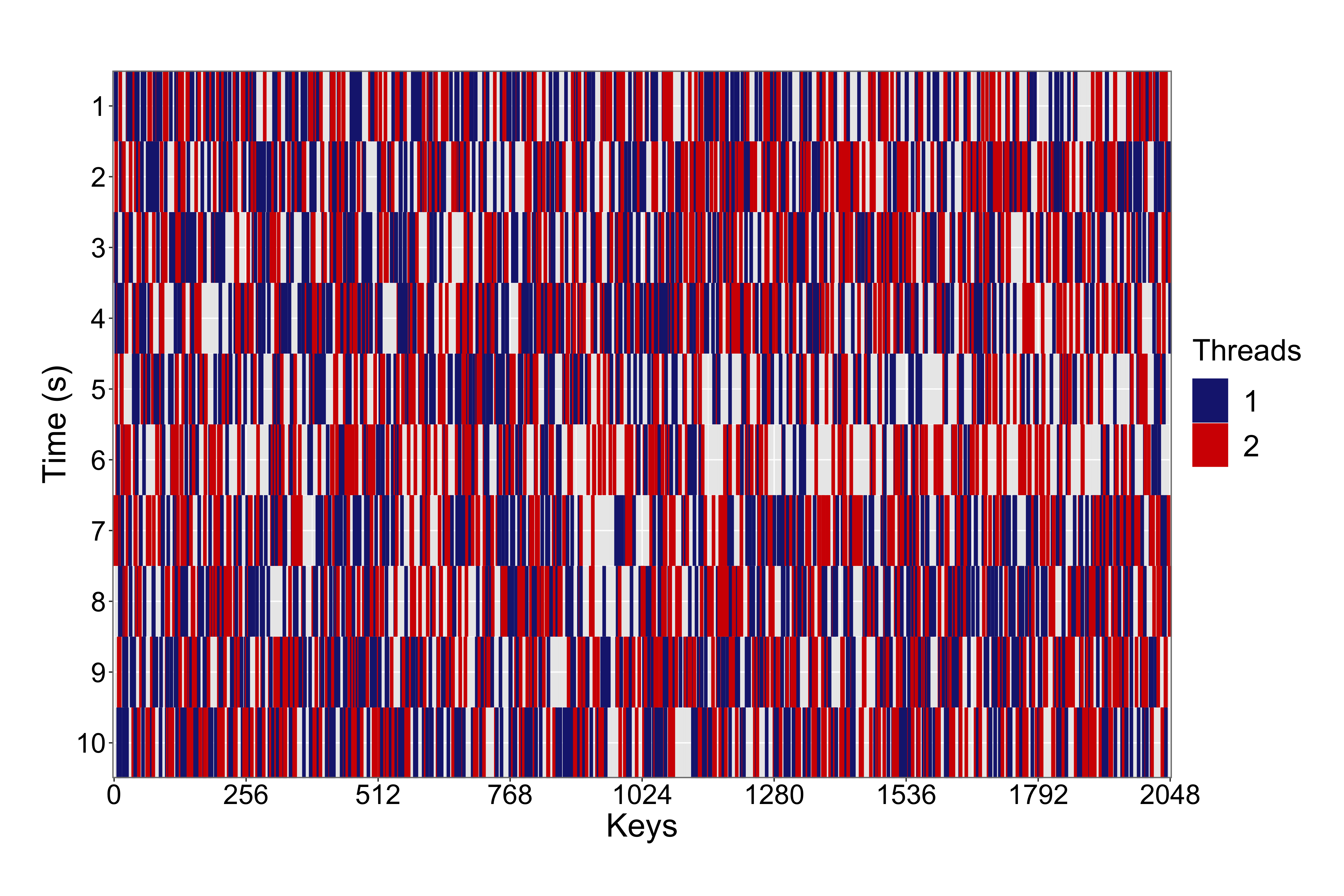}
		\caption{Key distribution, 2048 keyspace.}
		\label{fig:keydist-mc}
	\end{minipage}
\end{figure}


\section{Conclusion}

\label{Sec-Conclusion}

We presented a data partitioning technique based on using skip graphs that shows quantitative and qualitative improvements in NUMA locality, as well as reduced contention. We build on previous techniques of thread-local indexing and laziness, and, at a high level, our design consists of a partitioned skip graph, well-integrated with thread-local sequential maps, operating without contention.

As a proof-of-concept, we implemented map and relaxed priority queue ADTs using our technique. Maps were conceived using lazy and non-lazy approaches to insertions and removals. Remote memory accesses were not only reduced in number, but the larger the NUMA distance between threads, the larger the reduction was, demonstrating a qualitative feature besides a quantitative one. We provide an optional load-balancing mechanism, based on donating inserted nodes across thread-local indexes of all threads, coordinated by a single background thread. 

For relaxed priority queues, we consider two alternative implementations of relaxed priority queues: (a) using ``spraying'', a well-known random-walk technique usually performed over skip lists, but now performed over skip graphs; and (b) a custom protocol that traverses the skip graph deterministically, marking elements along this traversal. We provide formal arguments indicating that the first approach is slightly more \emph{relaxed}, that is, that the span of removed keys is larger, while the second approach has smaller contention. Experimental results indicate that the approach based on spraying performs better on skip graphs, yet both seem to scale appropriately.
%


\begin{thebibliography}{10}

\bibitem{skipnets}
Nicholas~J. A, Dunagan Michael, B.~Jones, and Stefan Saroiu.
\newblock Skipnet: {A} scalable overlay network with practical locality
  properties.
\newblock In {\em Proceedings of USENIX Symposium on Internet Technologies and
  Systems}, 2003.

\bibitem{frontaldatastructure2}
Deepthi Akkoorath, Jos{\'e} Brand{\~a}o, Annette Bieniusa, and Carlos Baquero.
\newblock Global-local view: Scalable consistency for concurrent data types.
\newblock In {\em Euro-Par 2018: Parallel Processing}, pages 492--504, Cham,
  2018. Springer International Publishing.

\bibitem{pqrelaxed1}
Dan Alistarh, Justin Kopinsky, Jerry Li, and Nir Shavit.
\newblock The spraylist: A scalable relaxed priority queue.
\newblock In {\em Proc. of the 20th ACM SIGPLAN Symposium on Principles and
  Practice of Parallel Programming}, PPoPP 2015, pages 11--20, New York, NY,
  USA, 2015. ACM.

\bibitem{RobinHood}
Martin Ankerl.
\newblock URL: \url{github.com/martinus/robin-hood-hashing}.

\bibitem{skipgraphs}
James Aspnes and Gauri Shah.
\newblock Skip graphs.
\newblock {\em ACM Transactions on Algorithms}, 3(4), November 2007.

\bibitem{delegation}
Irina Calciu, Dave Dice, Tim Harris, Maurice Herlihy, Alex Kogan, Virendra
  Marathe, and Mark Moir.
\newblock Message passing or shared memory: Evaluating the delegation
  abstraction for multicores.
\newblock In Roberto Baldoni, Nicolas Nisse, and Maarten van Steen, editors,
  {\em Principles of Distributed Systems}, pages 83--97. Springer International
  Publishing, 2013.

\bibitem{numastack}
Irina Calciu, Justin Gottschlich, and Maurice Herlihy.
\newblock Using elimination and delegation to implement a scalable
  {NUMA}-friendly stack.
\newblock In {\em 5th {USENIX} Workshop on Hot Topics in Parallelism}, San
  Jose, CA, 2013. {USENIX}.

\bibitem{pqe}
Irina Calciu, Hammurabi Mendes, and Maurice Herlihy.
\newblock The adaptive priority queue with elimination and combining.
\newblock In Fabian Kuhn, editor, {\em Distributed Computing}, volume 8784 of
  {\em Lec. Not. in Computer Science}, pages 406--420. Springer Berlin /
  Heidelberg, October 2014.

\bibitem{blackbox}
Irina Calciu, Siddhartha Sen, Mahesh Balakrishnan, and Marcos~K. Aguilera.
\newblock Black-box concurrent data structures for {NUMA} architectures.
\newblock {\em SIGPLAN Not.}, 52(4):207--221, April 2017.

\bibitem{nohotspot}
Tyler Crain, Vincent Gramoli, and Michel Raynal.
\newblock No hot spot non-blocking skip list.
\newblock In {\em Proceedings of the 2013 IEEE 33rd International Conference on
  Distributed Computing Systems}, ICDCS '13, pages 196--205, Washington, DC,
  USA, 2013. IEEE Computer Society.

\bibitem{numask}
Henry Daly, Ahmed Hassan, Michael~F. Spear, and Roberto Palmieri.
\newblock {NUMASK: High Performance Scalable Skip List for NUMA}.
\newblock In Ulrich Schmid and Josef Widder, editors, {\em 32nd International
  Symposium on Distributed Computing (DISC 2018)}, volume 121 of {\em Leibniz
  International Proceedings in Informatics (LIPIcs)}, pages 18:1--18:19,
  Dagstuhl, Germany, 2018. Schloss Dagstuhl--Leibniz-Zentrum fuer Informatik.

\bibitem{numaLock}
David Dice, Virendra~J. Marathe, and Nir Shavit.
\newblock Lock cohorting: A general technique for designing {NUMA} locks.
\newblock {\em SIGPLAN Not.}, 47(8):247--256, February 2012.

\bibitem{rotating}
Ian Dick, Alan Fekete, and Vincent Gramoli.
\newblock A skip list for multicore.
\newblock {\em Concurrency and Computation: Practice and Experience},
  29(4):e3876, 2017.

\bibitem{combining}
Panagiota Fatourou and Nikolaos~D. Kallimanis.
\newblock Revisiting the combining synchronization technique.
\newblock In {\em Proceedings of the 17th ACM SIGPLAN Symposium on Principles
  and Practice of Parallel Programming}, PPoPP '12, pages 257--266, New York,
  NY, USA, 2012. ACM.

\bibitem{ferrante2014coupon}
Marco Ferrante and Monica Saltalamacchia.
\newblock The coupon collector's problem.
\newblock {\em Materials matematics}, pages 0001--35, 2014.

\bibitem{skiplist-lockfree}
Mikhail Fomitchev and Eric Ruppert.
\newblock Lock-free linked lists and skip lists.
\newblock In {\em Proceedings of the Twenty-third Annual ACM Symposium on
  Principles of Distributed Computing}, PODC '04, pages 50--59, New York, NY,
  USA, 2004. ACM.

\bibitem{numaDB3}
Jana Giceva, Gustavo Alonso, Timothy Roscoe, and Tim Harris.
\newblock Deployment of query plans on multicores.
\newblock {\em Proceedings of the VLDB Endowement}, 8(3):233--244, November
  2014.

\bibitem{skipgraphs-rainbow}
Michael~T. Goodrich, Michael~J. Nelson, and Jonathan~Z. Sun.
\newblock The rainbow skip graph: A fault-tolerant constant-degree distributed
  data structure.
\newblock In {\em Proceedings of the Seventeenth Annual ACM-SIAM Symposium on
  Discrete Algorithm}, SODA '06, pages 384--393, Philadelphia, PA, USA, 2006.
  Society for Industrial and Applied Mathematics.

\bibitem{Synchrobench}
Vincent Gramoli.
\newblock More than you ever wanted to know about synchronization:
  Synchrobench, measuring the impact of the synchronization on concurrent
  algorithms.
\newblock In {\em Proceedings of the 20th ACM SIGPLAN Symposium on Principles
  and Practice of Parallel Programming}, PPoPP 2015, pages 1--10, New York, NY,
  USA, 2015. ACM.

\bibitem{flat-combining}
Danny Hendler, Itai Incze, Nir Shavit, and Moran Tzafrir.
\newblock Flat combining and the synchronization-parallelism tradeoff.
\newblock In {\em Proceedings of the Twenty-second Annual ACM Symposium on
  Parallelism in Algorithms and Architectures}, SPAA '10, pages 355--364, New
  York, NY, USA, 2010. ACM.

\bibitem{elimination}
Danny Hendler, Nir Shavit, and Lena Yerushalmi.
\newblock A scalable lock-free stack algorithm.
\newblock In {\em Proceedings of the Sixteenth Annual ACM Symposium on
  Parallelism in Algorithms and Architectures}, SPAA '04, pages 206--215, New
  York, NY, USA, 2004. ACM.

\bibitem{pqrelaxed3}
Thomas~A. Henzinger, Christoph~M. Kirsch, Hannes Payer, Ali Sezgin, and Ana
  Sokolova.
\newblock Quantitative relaxation of concurrent data structures.
\newblock In {\em Proceedings of the 40th Annual ACM SIGPLAN-SIGACT Symposium
  on Principles of Programming Languages}, POPL '13, pages 317--328, New York,
  NY, USA, 2013. Association for Computing Machinery.
\newblock URL: \url{https://doi.org/10.1145/2429069.2429109}.

\bibitem{TAOMP}
M.~Herlihy and N.~Shavit.
\newblock {\em The Art of Multiprocessor Programming}.
\newblock Morgan Kaufmann, 2008.

\bibitem{linearizability}
M.~Herlihy and J.~Wing.
\newblock Linearizability: a correctness condition for concurrent objects.
\newblock {\em ACM Transaction on Programming Languages and Systems},
  12:463--492, July 1990.

\bibitem{skiplist-optimistic}
Maurice Herlihy, Yossi Lev, Victor Luchangco, and Nir Shavit.
\newblock A simple optimistic skiplist algorithm.
\newblock In {\em Structural Information and Communication Complexity}, pages
  124--138, Berlin, Heidelberg, 2007. Springer Berlin Heidelberg.

\bibitem{progress}
Maurice Herlihy and Nir Shavit.
\newblock On the nature of progress.
\newblock In {\em Principles of Distributed Systems}, pages 313--328, Berlin,
  Heidelberg, 2011. Springer Berlin Heidelberg.

\bibitem{qdlocking}
David Klaftenegger, Konstantinos Sagonas, and Kjell Winblad.
\newblock Queue delegation locking.
\newblock {\em IEEE Transactions on Parallel and Distributed Systems},
  PP(99):1--1, 2017.

\bibitem{numaDB2}
Harald Lang, Viktor Leis, Martina-Cezara Albutiu, Thomas Neumann, and Alfons
  Kemper.
\newblock Massively parallel {NUMA}-aware hash joins.
\newblock In {\em In Memory Data Management and Analysis}, pages 3--14.
  Springer International Publishing, 2015.

\bibitem{numaDB4}
Viktor Leis, Peter Boncz, Alfons Kemper, and Thomas Neumann.
\newblock Morsel-driven parallelism: A {NUMA}-aware query evaluation framework
  for the many-core age.
\newblock In {\em Proceedings of the 2014 ACM SIGMOD International Conference
  on Management of Data}, SIGMOD '14, pages 743--754, New York, NY, USA, 2014.
  ACM.

\bibitem{skiplist-lazy}
Y.~Lev, M.~Herlihy, V.~Luchangco, and N.~Shavit.
\newblock A provably correct scalable skiplist (brief announcement).
\newblock In {\em Proceedings of the 10th International Conference On
  Principles Of Distributed Systems (OPODIS 2006)}, 2006.

\bibitem{skipgraphs-concurrent}
Hammurabi Mendes and Cristina~G. Fernandes.
\newblock A concurrent implementation of skip graphs.
\newblock {\em Electronic Notes in Discrete Mathematics}, 35:263--268, 2009.

\bibitem{cphash}
Zviad Metreveli, Nickolai Zeldovich, and M.~Frans Kaashoek.
\newblock Cphash: A cache-partitioned hash table.
\newblock In {\em Proceedings of the 17th ACM SIGPLAN Symposium on Principles
  and Practice of Parallel Programming}, PPoPP '12, pages 319--320, New York,
  NY, USA, 2012. ACM.

\bibitem{mitzenmacher2017probability}
Michael Mitzenmacher and Eli Upfal.
\newblock {\em Probability and computing: Randomization and probabilistic
  techniques in algorithms and data analysis}.
\newblock Cambridge university press, 2017.

\bibitem{skiplists-concpugh}
W.~Pugh.
\newblock Concurrent maintenance of skip lists.
\newblock Technical report, University of Maryland at College Park, 1990.

\bibitem{pqlotanshavit}
N.~Shavit and I.~Lotan.
\newblock Skiplist-based concurrent priority queues.
\newblock In {\em IEEE International Symposium on Parallel and Distributed
  Processing}, pages 263 --268, 2000.

\bibitem{pqsundelltsigas}
H.~Sundell and P.~Tsigas.
\newblock Fast and lock-free concurrent priority queues for multi-thread
  systems.
\newblock In {\em IEEE International Symposium on Parallel and Distributed
  Processing}, page 11 pp., april 2003.

\bibitem{phoenix}
Justin Talbot, Richard~M. Yoo, and Christos Kozyrakis.
\newblock Phoenix++: Modular {MapReduce} for shared-memory systems.
\newblock In {\em Proceedings of the Second International Workshop on
  {MapReduce} and Its Applications}, MapReduce '11, pages 9--16, New York, NY,
  USA, 2011. ACM.

\bibitem{numaDB1}
Mehul Wagle, Daniel Booss, Ivan Schreter, and Daniel Egenolf.
\newblock {NUMA}-aware memory management with in-memory databases.
\newblock In {\em Performance Evaluation and Benchmarking: Traditional to Big
  Data to Internet of Things}, pages 45--60. Springer International Publishing,
  2016.

\bibitem{pqrelaxed2}
Martin Wimmer, Jakob Gruber, Jesper~Larsson Tr\"{a}ff, and Philippas Tsigas.
\newblock The lock-free k-lsm relaxed priority queue.
\newblock In {\em Proceedings of the 20th ACM SIGPLAN Symposium on Principles
  and Practice of Parallel Programming}, PPoPP 2015, pages 277--278, New York,
  NY, USA, 2015. ACM.

\bibitem{pqrelaxed4}
Martin Wimmer, Francesco Versaci, Jesper~Larsson Tr\"{a}ff, Daniel Cederman,
  and Philippas Tsigas.
\newblock Data structures for task-based priority scheduling.
\newblock In {\em Proceedings of the 19th ACM SIGPLAN Symposium on Principles
  and Practice of Parallel Programming}, PPoPP '14, pages 379--380, New York,
  NY, USA, 2014. Association for Computing Machinery.
\newblock URL: \url{https://doi.org/10.1145/2555243.2555278}.

\end{thebibliography}

\newpage
\appendix

\section{Contains}
\label{App-Contains}

In algorithm \lstinline{contains()} (Alg.~\ref{alg:contains}), if the thread finds an unmarked, shared node, it will test if this node is unmarked and valid (line~\ref{alg:contains:getMarkValid}). If so, a successful contains is linearized right at that time (C-i). If not, we call \lstinline{SG::contains()} (Alg.~\ref{alg:sgContains}), which will complete the procedure.
\begin{algorithm}[htb]
\caption{bool Layered::contains(K key)}
\label{alg:contains}
\begin{algorithmic}[1]
	\State SharedNode result = hashtables[threadId].find(key)

	\If{result $\ne$ \textbf{null}}
		\If{\textbf{not} result.getMark(0)}
			\State \Return (result.getMarkValid(0) == (false, valid)) \label{alg:contains:getMarkValid}
		\Else
			\State localstructures[threadId].erase(key)
			\State hashtables[threadId].erase(key)
		\EndIf
	\EndIf 

	\State \Return SG::contains(key)
\end{algorithmic}
\end{algorithm}

The algorithm \lstinline{SG::contains()} (Alg.~\ref{alg:sgContains}) starts by calling \lstinline{getStart()} (Alg.~\ref{alg:getStart}), which finds the closest preceding, unmarked shared node pointed by the local structure (line~\ref{alg:lazyRemove:getStart}). We call this node \lstinline{currentStart}. We then perform a search in line~\ref{alg:contains:retireSearch}, starting from \lstinline{currentStart}, using our procedure \lstinline{retireSearch()} (Alg.~\ref{alg:retireSearch}). (C-ii) If an unmarked shared node with the goal key is not found, we linearize a failed contains operation at the time we found the element succeeding that key in the bottom level. (C-iii) If an unmarked shared node with the goal key is found: (C-iii-a) if we verify the node is unmarked and valid (test of line~\ref{alg:sgContains:getMarkValid}), we linearize a successful contains at that time. (C-iii-b) Otherwise, we linearize the failed contains at the earliest among the time we failed to verify the node was unmarked and valid (line~\ref{alg:sgContains:getMarkValid}), or right after the time the node became unmarked.
\begin{algorithm}[htb]
\caption{bool SG::contains(K key)}
\label{alg:sgContains}
\begin{algorithmic}[1]
	\CommentLine{Search the local structure and get the closest starting point}
	\State currentStart = getStart(key) \lstinline{alg:sgContains:getStart}

	\CommentLine{If an unmarked shared node is not found, the element cannot exist}
	\If{\textbf{not} retireSearch(key, found, currentStart)} \label{alg:contains:retireSearch}
		\State \Return \textbf{false}
	\EndIf

	\State \Return (result.getMarkValid(0) == (false, valid)) \label{alg:sgContains:getMarkValid}
\end{algorithmic}
\end{algorithm}

We now present our second search procedure, \lstinline{retireSearch()} (Alg.~\ref{alg:retireSearch}). This algorithm is employed by contains and remove operations to search for unmarked shared nodes with a goal \lstinline{key}, starting from \lstinline{currentStart}. The algorithm is a simplification of \lstinline{lazyRelinkSearch()} (Alg.~\ref{alg:lazyRelinkSearch}), as it does not keep track of predecessors or successors as the structure is traversed.
\begin{algorithm}[htb]
\caption{bool SG::retireSearch(K key, SharedNode found,\\ LocalStructureIterator currentStart)}
\label{alg:retireSearch}
\begin{algorithmic}[1]
	\State previous = currentStart.sharedNode
	\State current = \textbf{null}

	\For{level = MaxLevel $\rightarrow$ 0}
		\State current = previous.getNext(level)

		\While{current.getMark(0) \textbf{or} current.checkRetire())}
			\State current = current.getNext(level)
		\EndWhile

		\While{current.getKey() < key}
			\State previous = current
			\State current = originalCurrent = previous.getNext(level)

			\While{current.getMark(0) \textbf{or} current.checkRetire())}
				\State current = current.getNext(level)
			\EndWhile
		\EndWhile

		\If{current.getKey() == key \textbf{and not} current.getMark(0)}
			\State found = current
			\State \Return \textbf{true}
		\EndIf
	\EndFor

	\State \Return (successors[0].getKey() == key \textbf{and not} successors[0].getMark(0))
\end{algorithmic}
\end{algorithm}

\section{Extra Insertion Algorithms}
\label{App-ExtraInsertionAlgorithms}

The algorithm \lstinline{updateStart()} (Alg.~\ref{alg:updateStart}) finds an unmarked shared node pointed by the local structure that closest precedes the key associated with \lstinline{currentStart}.

Starting from \lstinline{currentStart}, it traverses the local structure backwards (line~\ref{alg:updateStart:traverseBackwards}) for as long as the shared nodes pointed by the local structure are found marked. Note that \lstinline{updateStart()} is essentially a simplified version of \lstinline{getStart()} that does not finish insertions lazily, and ignores nodes that are not fully inserted as it traverses the local structure backwards.
\begin{algorithm}[htb]
\caption{void updateStart(\textbf{ref} LocalStructureIterator iterator)}
\label{alg:updateStart}
\begin{algorithmic}[1]
	\While{iterator.sharedNode $\ne$ \textbf{null}}
		\State SharedNode sharedNode = iterator.sharedNode
		\If{\textbf{not} sharedNode.getMark(0) \textbf{or} \textbf{not} sharedNode.getMark(MaxLevel)}
			\If{\textbf{not} sharedNode.inserted}
				\State \Return
			\EndIf
		\Else
			\CommentLine{Erase below does not invalidate the iterator}
			\State localstructures[threadId].erase(key)
			\State hashtables[threadId].erase(key)
		\EndIf
		\State result = result.getPrev() \label{alg:updateStart:traverseBackwards}
	\EndWhile
\end{algorithmic}
\end{algorithm}

The algorithm \lstinline{finishInsert()} (Alg.~\ref{alg:finishInsert}) links elements at levels $1 \ldots \mathrm{MaxLevel}$ in the thread's associated skip list. The procedure uses at line~\ref{alg:finishInsert:lazyRelinkSearch} the \lstinline{lazyRelinkSearch()} procedure to locate, at all levels $0 \ldots \mathrm{MaxLevel}$: (A) which nodes should precede \lstinline{toInsert} (referenced in \lstinline{predecessors}); (B) which nodes should succeed \lstinline{toInsert} (referenced in \lstinline{successors}); and (C) nodes referenced in \lstinline{middle} such that \lstinline{middle[i]} contains \lstinline{predecessors[i].getNext(i)} right after the time each predecessor \lstinline{predecessor[i]} was identified. Between \lstinline{predecessors[i]} and \lstinline{successors[i]} we have a sequence of nodes with their level-$i$ references marked. We hope that in the CAS of line~\ref{alg:finishInsert:CAS}, we replace that sequence of nodes with \lstinline{toInsert}. If the CAS fails, we update \lstinline{currentStart} with \lstinline{updateStart()}, and retry the search. If at any moment, we cannot locate an unmarked \lstinline{toInsert}, this node has been marked while this procedure was taking place, and we indicate our failure to build all levels in line~\ref{alg:finishInsert:fail}.
\begin{algorithm}[htb]
\caption{bool finishInsert(SharedNode toInsert, LocalStructureIterator currentStart)}
\label{alg:finishInsert}
\begin{algorithmic}[1]
	\State key = toInsert.key
	\State Array predecessors[MaxLevel], middle[MaxLevel], successors[MaxLevel]

	\While{\textbf{true}}
		\If{\textbf{not} SG::lazyRelinkSearch(key, predecessors, middle, successors, currentStart)} \label{alg:finishInsert:lazyRelinkSearch}
			\State \Return \textbf{false} \label{alg:finishInsert:fail}
		\EndIf

		\For{level = 1 $\rightarrow$ MaxLevel}
			\State SharedNode oldSuccessor = toInsert.getNext(level)	

			\While{\textbf{not} toInsert.casNext(level, oldSuccessor, successors[level])}
				\State SharedNode oldSuccessor = toInsert.getNext(level)

				\If{toInsert.getMark(level)} \Comment{If marked, abort linking}
					\State toInsert.inserted = true
					\State \Return \textbf{false}
				\EndIf
			\EndWhile

			\If{\textbf{not} predecessors[level].casNext(level, middle[level], toInsert)} \label{alg:finishInsert:CAS}
				\State SG::lazyRelinkSearch(key, predecessors, middle, successors, currentStart)
				\State updateStart(currentStart)
				\State level = level - 1 \Comment{Retry the level}
			\EndIf
		\EndFor

		\State toInsert.inserted = true
		\State \Return \textbf{true}
	\EndWhile
\end{algorithmic}
\end{algorithm}

\section{Remove}
\label{App-Remove}

Algorithms~\ref{alg:remove}, \ref{alg:removeHelper}, and~\ref{alg:lazyRemove} show the remove operation for \lstinline{key}. In \lstinline{remove()} (Alg.~\ref{alg:remove}), if the thread finds a reference to a shared node with the goal key, it calls \lstinline{removeHelper()} (Alg.~\ref{alg:removeHelper}). This call returns true only if: (R-i) We found an unmarked, invalid node, so we linearize a failed removal right at that time. There cannot exist another unmarked, valid node in this case, which justifies our linearizability. (R-ii) We successfully unset the \lstinline{valid} bit of an unmarked, valid node, so we linearize the successful removal right at that time as well. If \lstinline{removeHelper()} returned false, the shared node previously found is marked, so (i) we clean the thread's local structure (Alg.~\ref{alg:removeHelper}, lines~\ref{alg:removeHelper:cleanup1} and~\ref{alg:removeHelper:cleanup2}); and (ii) we call \lstinline{lazyRemove()} (Alg.~\ref{alg:lazyRemove}), which will complete the removal.
\begin{algorithm}[htb]
\caption{bool Layered::remove(T key, U value)}
\label{alg:remove}
\begin{algorithmic}[1]
	\State SharedNode result = hashtables[threadId].find(key)

	\If{result $\ne$ \textbf{null}}
		\State bool returnValue = false
		\If{SG::removeHelper(result, returnValue)}
			\State \Return returnValue
		\EndIf
	\EndIf 

	\State \Return SG::remove(key)
\end{algorithmic}
\end{algorithm}

\begin{algorithm}[htb]
\caption{bool SG::removeHelper(SharedNode toRemove, \textbf{ref} bool returnValue)}
\label{alg:removeHelper}
\begin{algorithmic}[1]
	\While{\textbf{true}}
		\If{\textbf{not} toRemove.getMark(0)}
			\If{toRemove.getMarkValid(0) == (false, invalid)} \Comment{Non-existent}
				\State bool returnValue = \textbf{false}
				\State \Return \textbf{true}
			\EndIf
			\If{toRemove.casMarkValid((false, valid), (false, invalid))} \Comment{Flipped \lstinline{valid}}
				\State returnValue = \textbf{true}
				\State \Return \textbf{true}
			\EndIf
		\Else
			\State localstructures[threadId].erase(toRemove.getKey()) \label{alg:removeHelper:cleanup1}
			\State hashtables[threadId].erase(toRemove.getKey()) \label{alg:removeHelper:cleanup2}

			\State \textbf{break}
		\EndIf
	\EndWhile

	\State \Return \textbf{false}
\end{algorithmic}
\end{algorithm}

The algorithm \lstinline{lazyRemove()} (Alg.~\ref{alg:lazyRemove}) starts by calling \lstinline{getStart()} (Alg.~\ref{alg:getStart}), which finds the closest preceding, unmarked shared node pointed by the local structure (line~\ref{alg:lazyRemove:getStart}). We call this node \lstinline{currentStart}. We then perform a search in line~\ref{alg:lazyRemove:retireSearch}, starting from \lstinline{currentStart} and: (R-iii) If an unmarked shared node with the goal key is found (call it \lstinline{found}), then a call to \lstinline{removeHelper()} (line~\ref{alg:lazyRemove:removeHelper}) will define our linearization, just like our previous discussion of Alg.~\ref{alg:remove}. If that call returns false, \lstinline{found} became marked, so we retry the search. (R-iv) No unmarked node with the goal key has been found, so in that case we linearize a failed removal at the time we found the element succeeding that key in the bottom level.
\begin{algorithm}[htb]
\caption{bool SG::lazyRemove(SharedNode current)}
\label{alg:lazyRemove}
\begin{algorithmic}[1]
	\CommentLine{Search the local structure and get the closest starting point}
	\State currentStart = getStart(key) \label{alg:lazyRemove:getStart}

	\While{\textbf{true}}
		\CommentLine{If an unmarked shared node is not found, the element cannot exist}
		\If{\textbf{not} retireSearch(key, toRemove, currentStart)} \label{alg:lazyRemove:retireSearch}
			\State \Return \textbf{false}
		\EndIf

		\State returnValue = false
		\If{SG::removeHelper(successors[0], returnValue)} \label{alg:lazyRemove:removeHelper}
			\State \Return returnValue
		\EndIf
	\EndWhile
\end{algorithmic}
\end{algorithm}

\subsection{Retiring Nodes}
\label{App-RemoveRetire}

The algorithm \lstinline{checkRetire()} (Alg.~\ref{alg:checkRetire}) is performed as we search for elements on behalf of \lstinline{insert()} or \lstinline{remove()}. The procedure returns true only if a node is unmarked and its commission period has expired, indicating that the node a candidate for physical removal. A subsequent call to \lstinline{retire()} will \emph{try} to start the physical removal process over a node, if the latter is found invalid.
\begin{algorithm}[htb]
\caption{bool SG::checkRetire(SharedNode current)}
\label{alg:checkRetire}
\begin{algorithmic}[1]
	\State isAliveInvalid = (current.getMarkValid(0) == (\textbf{false}, invalid)
	\State isCommissionFinished = (timestamp() - current.allocTimestamp) > CommissionPeriod
	\If{isAliveInvalid \textbf{and} isCommissionFinished}
		\State \Return retire(current)
	\EndIf
	\State \Return \textbf{false}
\end{algorithmic}
\end{algorithm}

The algorithm \lstinline{retire()} (Alg.~\ref{alg:retire}) tries to atomically invalidate a node's bottom level reference, while guaranteeing that this reference was invalid throughout the process (line~\ref{alg:retire:CAS}). The CAS performed at such line can fail because (i) a node's reference has been changed; or (ii) because the node became just before the CAS. In either situation, our CAS will fail and we simply return false. In an active node, particularly under case (ii), we would prefer to \emph{keep} the node unmarked, so returning false has an operational advantage.
\begin{algorithm}[htb]
\caption{bool SG::retire(SharedNode current)}
\label{alg:retire}
\begin{algorithmic}[1]
	\If{\textbf{not} current.casMarkValid(0, (\textbf{false}, invalid), (\textbf{true}, invalid))} \label{alg:retire:CAS}
		\State \Return \textbf{false}
	\EndIf

	\For{level = MaxLevel $\rightarrow$ 1}
		\While{\textbf{not} current.getMark(level)}
			\State current.casMark(level, \textbf{false}, \textbf{true})
		\EndWhile
	\EndFor

	\State \Return \textbf{true}
\end{algorithmic}
\end{algorithm}

\section{More Experimental Results: LC tests}
\label{App-LCTests}

We believe the LC-WH scenario is very meaningful from an analytic standpoint. First, it shows that our layered approaches take full advantage of a low contention setting, matching the performance of a locked skip list (expected to work very well in such scenarios).
 
We attribute this to (i) highly effective local structures, with no overhead of atomic operations or locks; (ii) to the fact that new nodes need to be inserted in only 2 levels on expectation, instead of a fixed 2-6, depending on the number of threads, as prescribed by our partitioning policy; (iii) to a better cache efficiency, documented in Tbl.~\ref{tbl:cache_misses}; and (iv) to better NUMA locality (verified in item \#2, below). The performance of \texttt{layered\_map\_ssg} in fact reveals an interesting tradeoff: although less shared nodes reach the topmost level and get included in the local structures, thus making the starting point of shared operations not as ``close'' as in the non-sparse version, the number of CAS operations required to insert nodes is substantially lower than \texttt{layered\_map\_sg}. The poor performance of non-layered skip graphs also reflects a higher number of required CAS operations for insertion.

Another reason for why the LC-WH scenario appears meaningful is that it allows us to visualize the lower performance of \texttt{lazy\_layered\_map\_sg} as compared to \texttt{layered\_map\_sg}, what we believe to be a direct consequence of our conservative commission policy to unlink invalid, marked nodes (p.~\pageref{commission}). We think this happens as the commission policy (essentially, a lazy unlink policy) is making the shared structure contain more invalid or marked nodes present in the LC-WH scenario compared to the HC-WH or MC-WH scenarios. In any case, we do think that the commission policy overhead is a fair price for the highly increased performance in the MC and HC scenarios, yet still allowing us to outperform competitor maps in the LC-WH as well. A better NUMA locality, and decreased contention (verified in item \#2, below) seems to be making up for the overhead in the LC-WH scenario. In the MC-WH scenario, we were very impressed with the performance of the \texttt{nohotspot} approach, noting that we just barely outperformed it. This further indicates that the commission policy might be subject to having a ``sweet spot'', and thus our performance might be responsive to more elaborated commission policy designs, an interesting prospect for future work.

\begin{figure}[htb]
    \centering
	\begin{minipage}{0.48\textwidth}
		\centering
		\includegraphics[width=\textwidth]{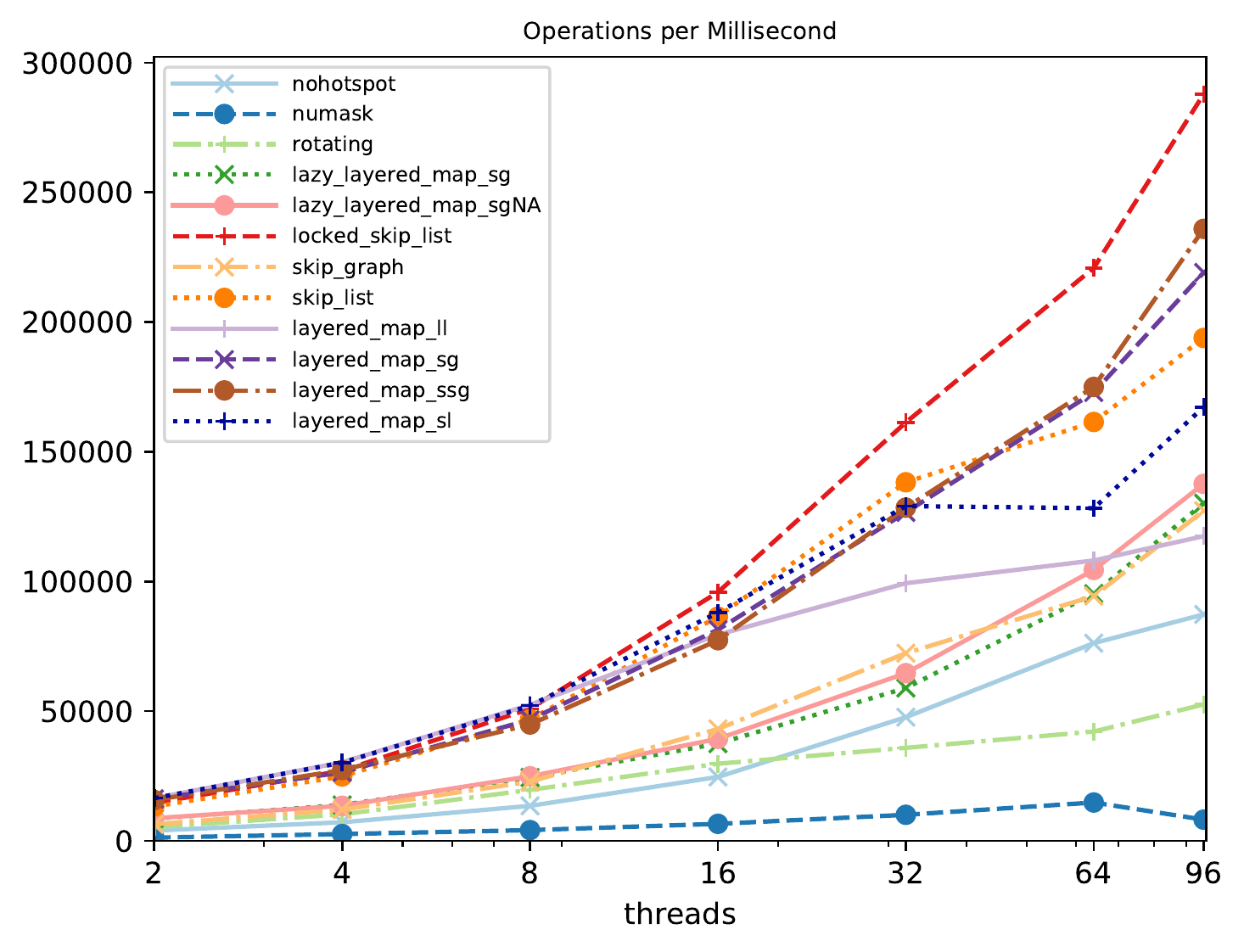}
		\caption{LC, WH: 4\% effective updates.}
		\label{fig:performance-lc-wh}
	\end{minipage}\hfill
\end{figure}

\section{More Experimental Results: RH tests}
\label{App-RHTests}

We present here results analogous to Figs.~\ref{fig:performance-hc-wh}, \ref{fig:performance-mc-wh}, and \ref{fig:performance-lc-wh}, but showing performance on RH scenarios. Please refer to the discussion of p. \pageref{discussion:performance}.

\begin{figure}[H]
    \centering
	\begin{minipage}{0.48\textwidth}
		\centering
		\includegraphics[width=\textwidth]{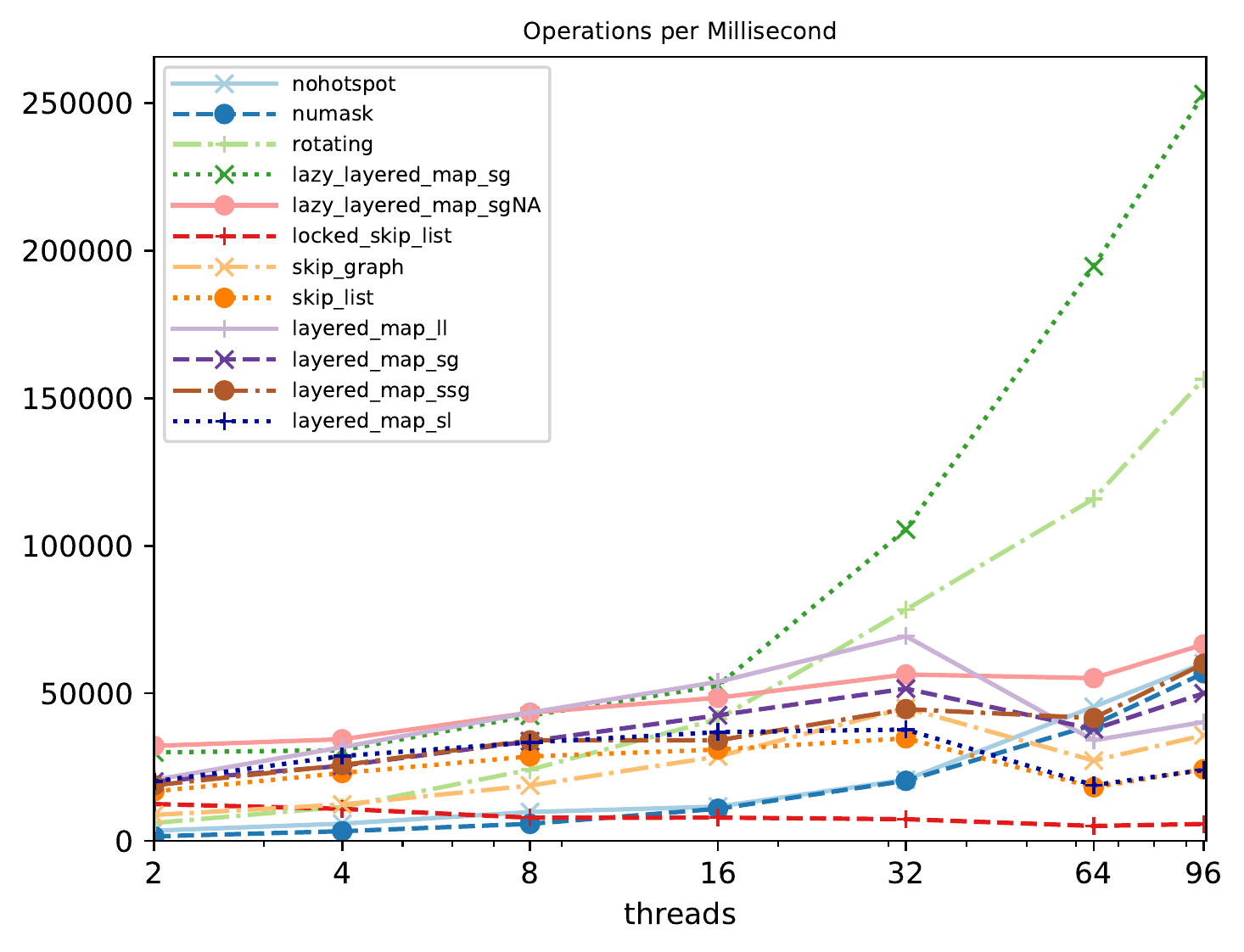}
		\caption{HC, RH: 32\% effective updates.}
		\label{fig:performance-hc-rh}
	\end{minipage}\hfill
	\begin{minipage}{0.48\textwidth}
		\centering
		\includegraphics[width=\textwidth]{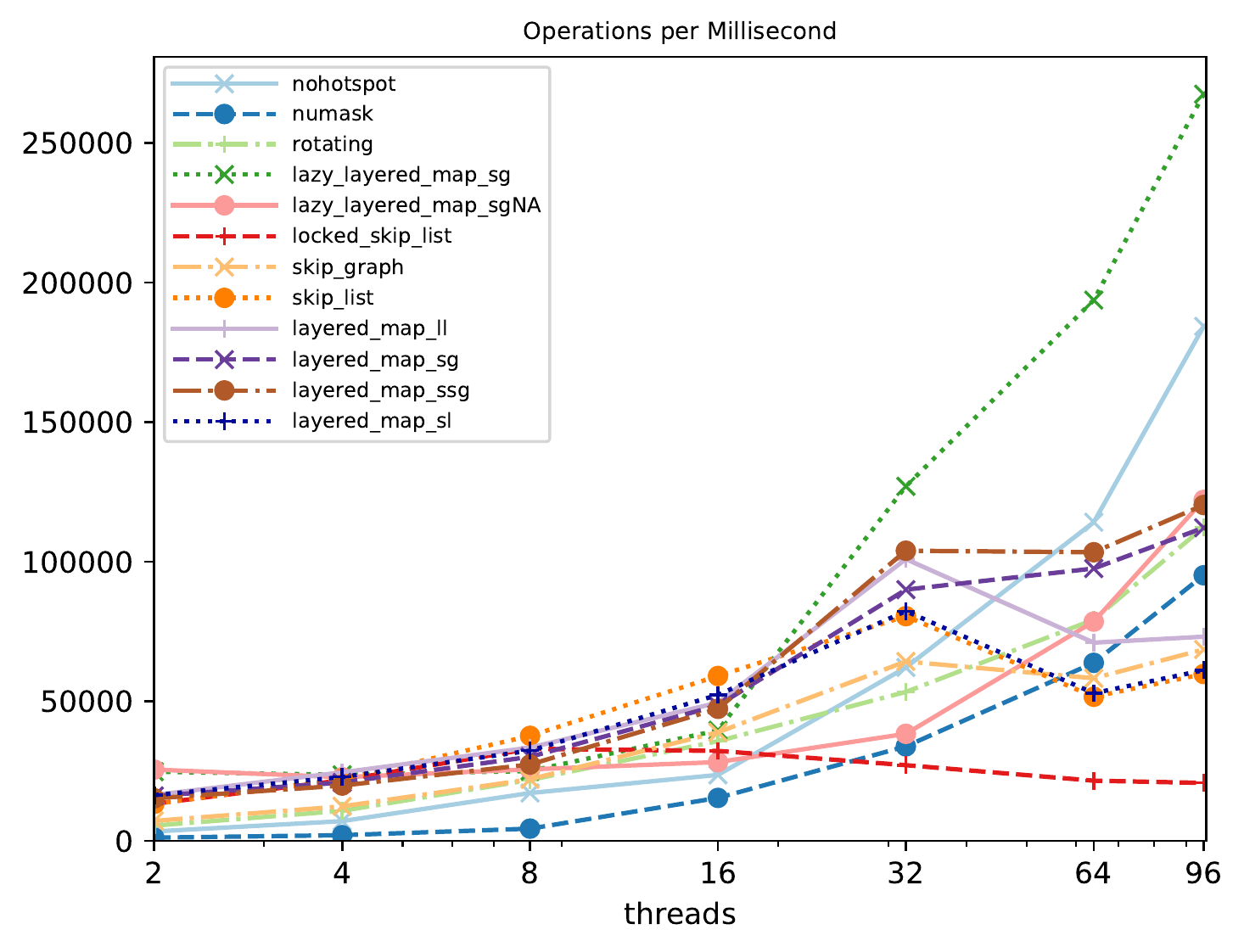}
		\caption{MC, RH: 32\% effective updates.}
		\label{fig:performance-mc-rh}
	\end{minipage}
\end{figure}

\begin{figure}[H]
    \centering
	\begin{minipage}{0.48\textwidth}
		\centering
		\includegraphics[width=\textwidth]{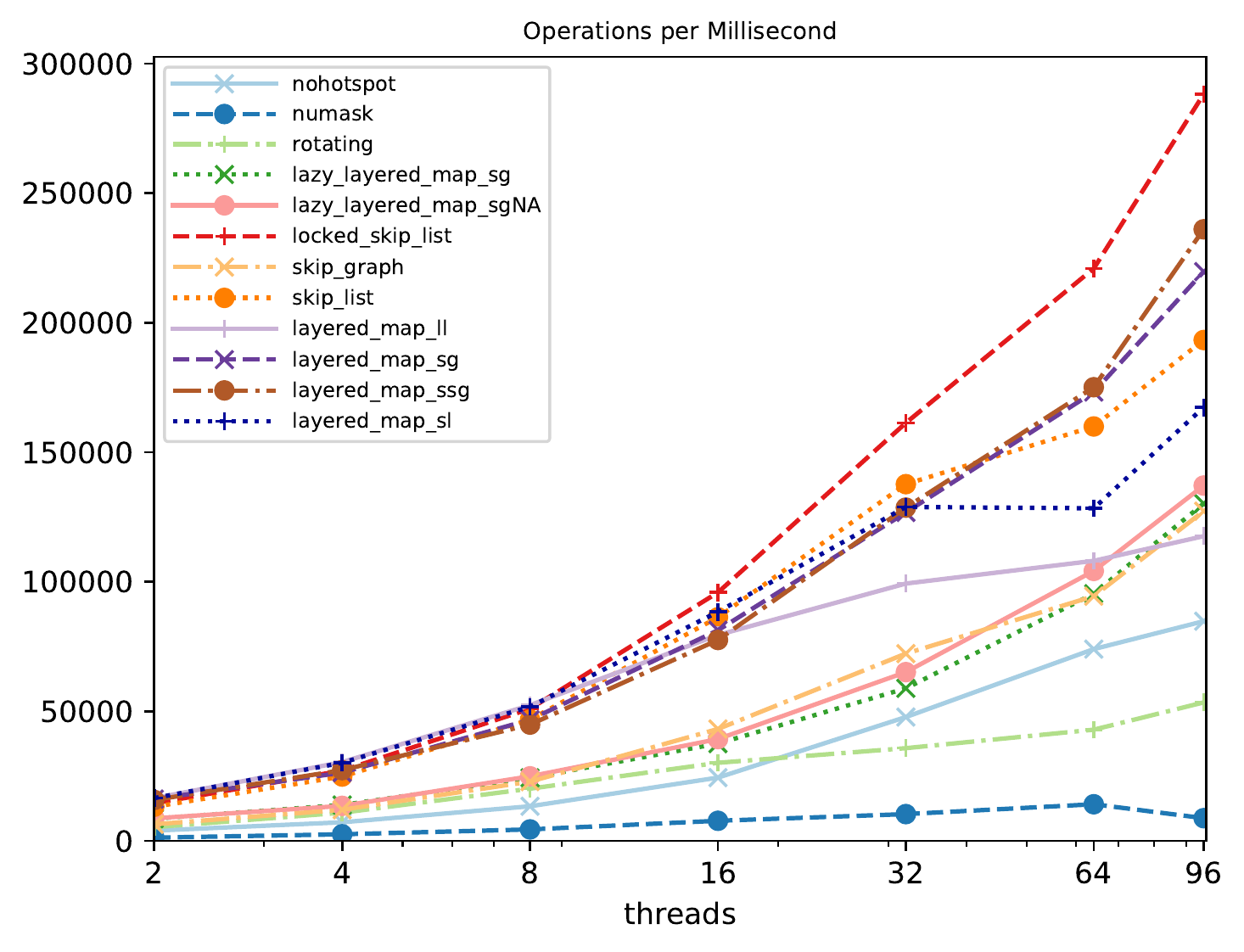}
		\caption{LC, RH: 4\% effective updates.}
		\label{fig:performance-lc-rh}
	\end{minipage}\hfill
\end{figure}

\section{More Experimental Results: Locality Heatmaps}
\label{App-LocalityHeatmaps}

We present here results analogous to Figs.\ref{fig:heatmap-cas-lazysg} and~\ref{fig:heatmap-cas-sl}, but showing read locality heatmaps. Please refer to the discussion of p. \pageref{discussion:locality}. The graphics presented here are not normalized by the total number of CAS operations. So, for example, the difference in scale between our lazy layered skip graph and a skip list, results from the former performing about 8x the number of operations of the latter.

We note that in Fig.~\ref{fig:heatmap-read-ssg}, the vertical line at thread 1 represents accesses to the \lstinline{head} array of the shared structure, which we arbitrarily attributed to that thread. We expect (and verify) that \lstinline{head} accesses are more common with \texttt{layered\_map\_ssg} as compared to \texttt{layered\_map\_sg}, due to the smaller size of the shared structure.

\begin{figure}[H]
    \centering
	\begin{minipage}{0.48\textwidth}
		\centering
		\includegraphics[width=\textwidth]{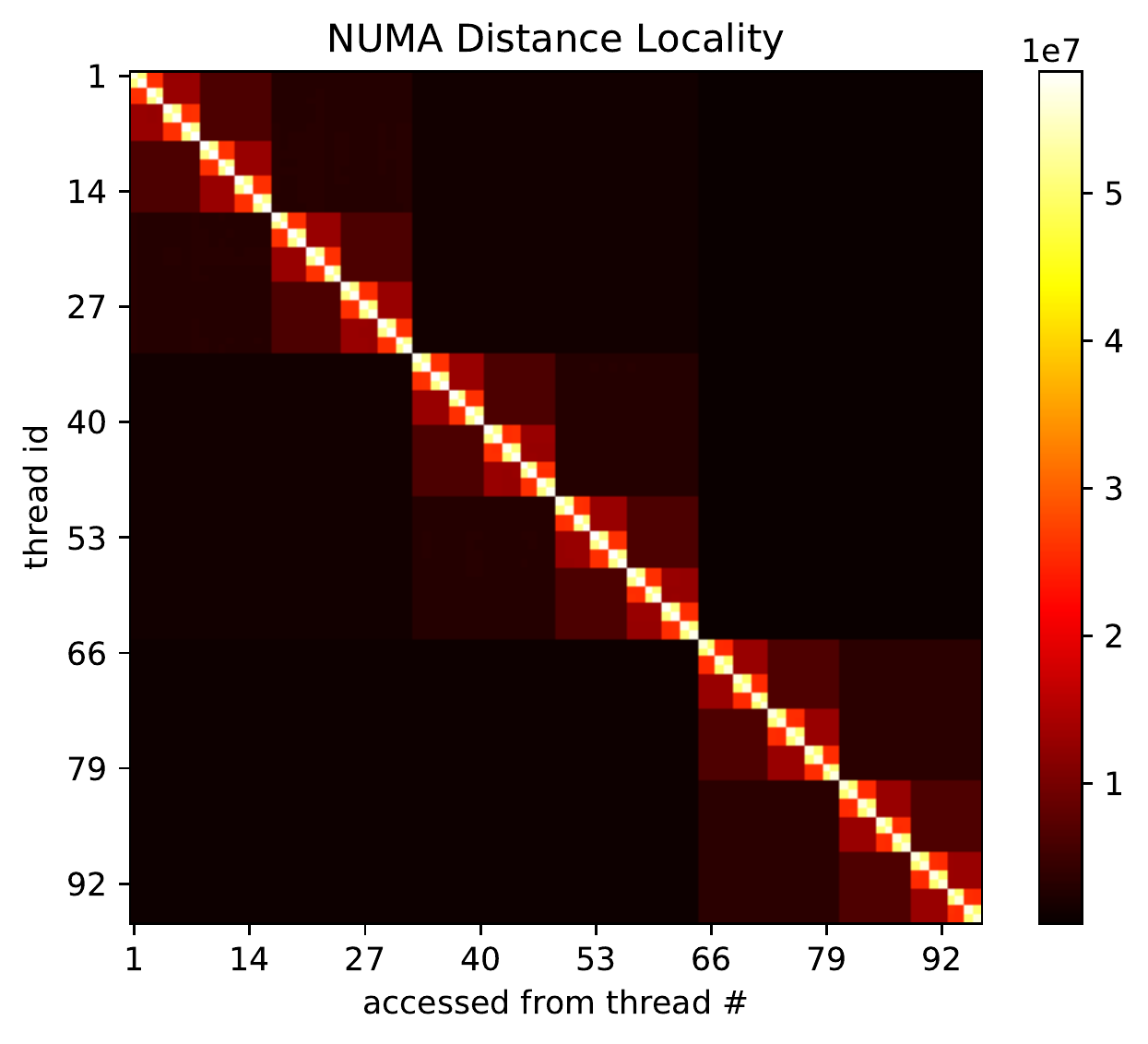}
		\caption{Read heatmap, lazy map/SG.}
		\label{fig:heatmap-read-lazysg}
	\end{minipage}\hfill
	\begin{minipage}{0.48\textwidth}
		\centering
		\includegraphics[width=\textwidth]{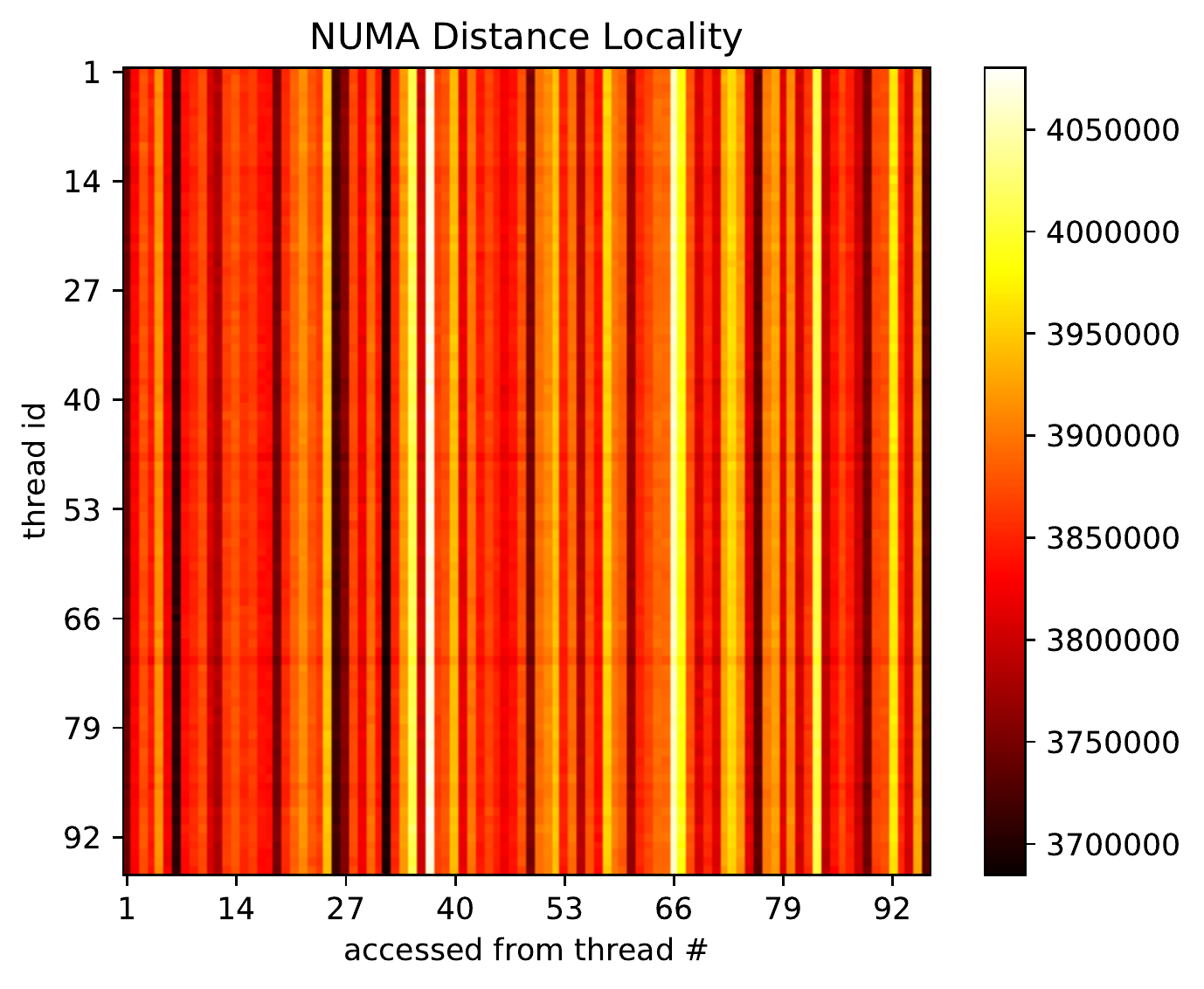}
		\caption{Read heatmap, skip list.}
		\label{fig:heatmap-read-sl}
	\end{minipage}
\end{figure}

\begin{figure}[H]
    \centering
	\begin{minipage}{0.48\textwidth}
		\centering
		\includegraphics[width=\textwidth]{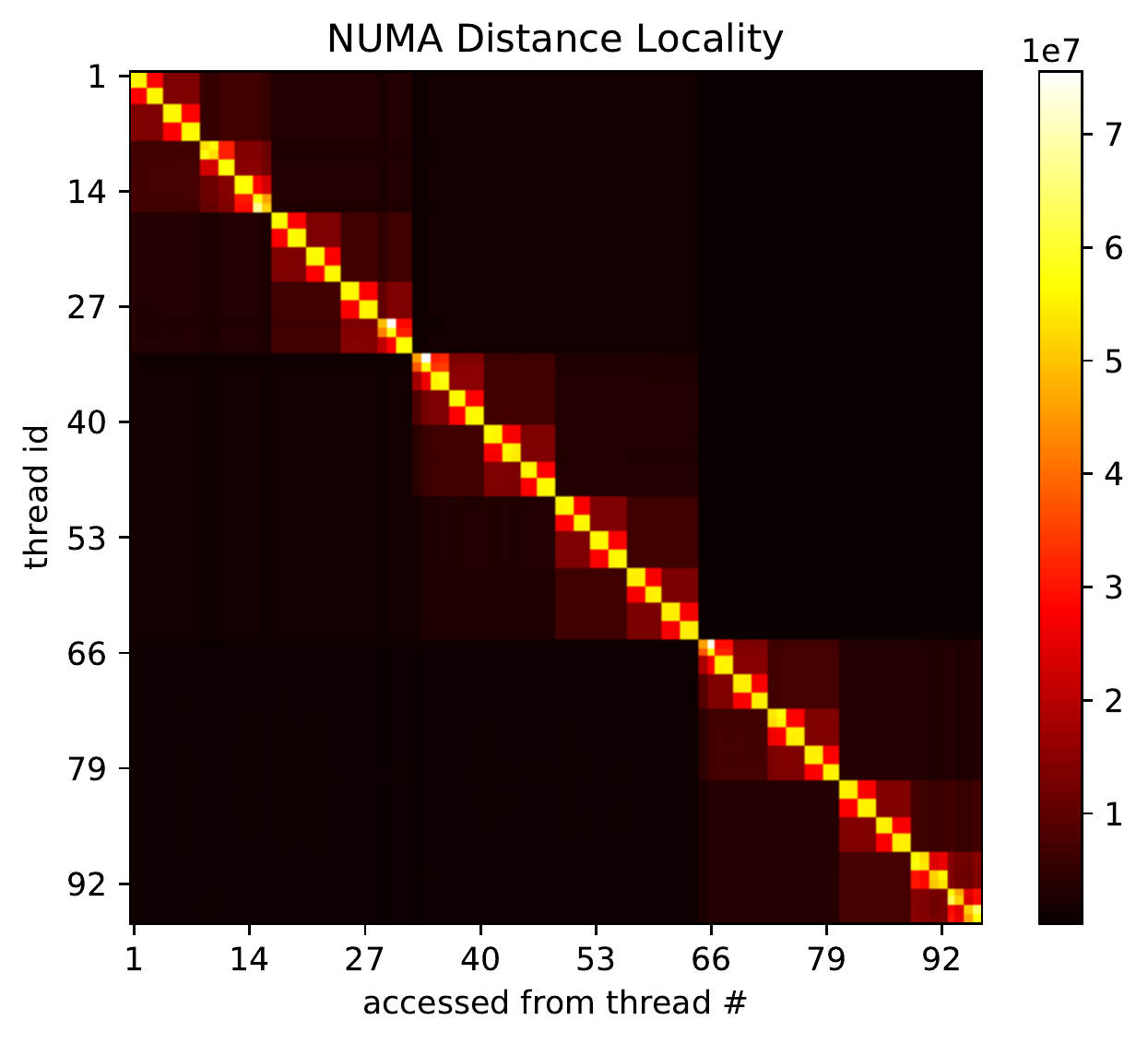}
		\caption{Read heatmap, non-lazy map/SG.}
		\label{fig:heatmap-read-sg}
	\end{minipage}\hfill
	\begin{minipage}{0.48\textwidth}
		\centering
		\includegraphics[width=\textwidth]{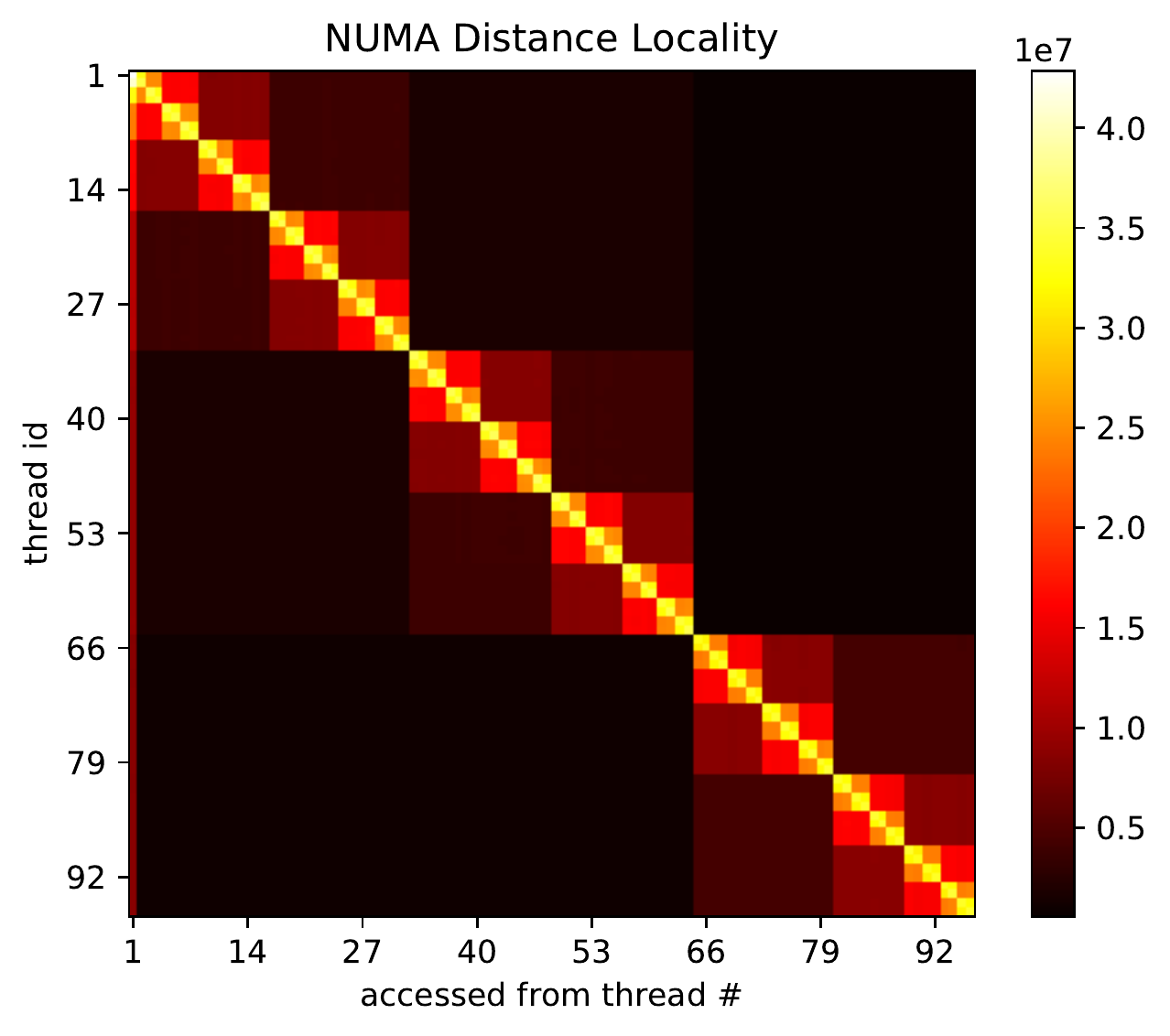}
		\caption{Read heatmap, sparse skip graph.}
		\label{fig:heatmap-read-ssg}
	\end{minipage}
\end{figure}

\end{document}